\def\year{2017}\relax
\documentclass[letterpaper]{article}
\usepackage{aaai17}
\usepackage{times}
\usepackage{helvet}
\usepackage{courier}

\usepackage{url}
\usepackage{amsmath}
\usepackage{amsthm}
\usepackage{amsfonts}
\usepackage{amssymb}
\usepackage{mathtools}
\usepackage{algorithm}
\usepackage{algpseudocode}
\usepackage{algorithmicx}
\usepackage{multirow}
\usepackage[inline]{enumitem}
\usepackage{tabularx}
\usepackage{subcaption}
\usepackage{etoolbox}

\usepackage{tikz}
\usetikzlibrary{graphs}
\usetikzlibrary{positioning,calc}

\newtheorem{myclaim}{Claim}
\newtheorem{mylemma}{Lemma}

\newtheorem{mythm}{Theorem}
\newtheorem{myprop}{Proposition}
\newtheorem{myeg}{Example}
\AtEndEnvironment{myeg}{\null\hfill\qedsymbol}
\let\oldeg\myeg
\renewcommand{\myeg}{\oldeg\normalfont}
\newtheorem{myobs}{Observation}
\let\oldobs\myobs
\renewcommand{\myobs}{\oldobs\normalfont}
\newtheorem{mydef}{Definition}
\let\olddef\mydef
\renewcommand{\mydef}{\olddef
}

\makeatletter
\newcommand{\algmargin}{\the\ALG@thistlm}
\makeatother
\newlength{\whilewidth}
\algdef{SE}[parWHILE]{parWhile}{EndparWhile}[1]
  {\parbox[t]{\dimexpr\linewidth-\algmargin}{%
     \hangindent\whilewidth\strut\algorithmicwhile\ #1\ \algorithmicdo\strut}}{\algorithmicend\ \algorithmicwhile}%
\algnewcommand{\parState}[1]{\State%
  \parbox[t]{\dimexpr\linewidth-\algmargin}{\strut #1\strut}}

\newcommand{\citeay}[1]{\citeauthor{#1}~\shortcite{#1}}

\newcommand\Omit[1]{}

\newcommand{\mttc}{\text{MTTC}}
\newcommand{\mttcs}{\text{MTTC*}}
\newcommand{\ma}{{\mathcal A}}
\newcommand{\mc}{{\mathcal C}}

\newcommand{\mt}{{\mathcal T}}
\newcommand{\mn}{{\mathcal N}}
\newcommand{\mo}{{\mathcal O}}
\newcommand{\order}{\text{Order}}
\newcommand{\po}{\text{PO}}
\newcommand{\type}{\text{Type}}
\newcommand{\cyc}{\text{Cycles}}
\newcommand{\ext}{\text{Ext}}
\newcommand{\pa}{\text{Pa}}
\newcommand{\cpt}{\text{CPT}}

\begin{document}
%
\title{Mechanism Design for Multi-Type Housing Markets}
\author{
Sibel Adal{\i} \\ Department of Computer Science \\ Rensselaer Polytechnic Institute \\ sibel@cs.rpi.edu \\
\And Sujoy Sikdar \\ Department of Computer Science \\ Rensselaer Polytechnic Institute \\ sikdas@rpi.edu \\
\And Lirong Xia \\ Department of Computer Science \\ Rensselaer Polytechnic Institute \\ xial@cs.rpi.edu\\
}
\maketitle
\begin{abstract}

We study {\em multi-type housing markets}, where there are $p\ge 2$ types of items, each agent is initially endowed one item of each type, and the goal is to design mechanisms without monetary transfer to (re)allocate items to the agents based on their preferences over {\em bundles} of items, such that each agent gets one item of each type. In sharp contrast to classical housing markets, previous studies in multi-type housing markets have been hindered by the lack of natural solution concepts, because the strict core might be empty.

We break the barrier in the literature by leveraging AI techniques and making natural assumptions on agents' preferences. We show that when agents' preferences are lexicographic, even with different importance orders, the classical top-trading-cycles mechanism can be extended while preserving most of its nice properties. We also investigate computational complexity of checking whether an allocation is in the strict core and checking whether the strict core is empty. Our results convey an encouragingly positive message: {\em it is possible to design good mechanisms for multi-type housing markets under natural assumptions on preferences.}

\end{abstract}

\section{Introduction}

In this paper, we ask the following question: {\em is it possible at
  all to design good mechanisms for multi-type housing markets?}  In
multi-type housing
markets~\cite{Moulin95:cooperative,Konishi01:On-the-ShapleyScarf,Wako2005:Coalition-proof,Klaus2008:The-coordinate-wise},
there are multiple types of items, each agent is initially endowed one
item of each type.  The goal is to design mechanisms without monetary
transfer to (re)allocate items to the agents based on their
preferences over {\em bundles} of items, such that each agent gets one
item of each type. 

Multi-type housing markets are often described using
examples of houses and cars as metaphors for indivisible
items. However, the allocation problem is applicable to many other
types of items and scarce resources.  For example, students may want
to exchange papers and dates for
presentation~\cite{Mackin2016:Allocating}; in cloud computing, agents
may want to allocate multiple types of resources, including CPU,
memory, and storage~\cite{Ghodsi11:Dominant,Ghodsi12:Multi}; patients
may want to allocate multiple types of medical resources, including
surgeons, nurses, rooms, and equipments~\cite{Huh13:Multiresource}.

Mechanism design for single-type housing markets is a well-established
field in economics, often referred to as {\em housing
  markets}~\cite{Shapley74:Cores}. In housing markets, the most
sensible solution concept is the {\em strict core}, which is the set
of allocations where no group of agents have incentive to deviate by
exchanging their initial endowments within the group. Strict core is desirable because it is an intuitive stable solution, and when agents' preferences are linear orders, the strict core
allocation always exists and is unique, which can be
computed in polynomial time by Gale's celebrated {\em
  Top-Trading-Cycles (TTC)}
algorithm~\cite{Shapley74:Cores,Roth77:Weak,Abdulkadiroglu99:House}. TTC
enjoys many desirable axiomatic properties including {\em individual
  rationality}, {\em Pareto optimality}, and {\em
 group strategy-proofness}. Many extensions of TTC to other single-type
housing markets have been proposed and studied. See more details in
Related Work.

In sharp contrast to the popularity of housing markets, there is
little research on multi-type housing markets, despite their
importance and generality. A potential reason for the absence of
positive results is that the strict core can be empty or multi-valued
in multi-type housing
markets~\cite{Konishi01:On-the-ShapleyScarf}. Therefore,
as~\citeay{Sonmez11:Matching} noted: {\em ``Positive results of this
  section [on housing markets] no longer hold in an economy in which
  one agent can consume multiple houses or multiple types of
  houses''}. This is the problem we address in this paper and provide
a number of positive results, a first in this field.

\subsection{Our contributions}
In this paper, we present novel algorithms building on AI techniques
in preference representation and reasoning for allocation in
multi-type housing markets. We assume that agents'
preferences are represented by arbitrary acyclic {\em
  CP-nets}~\cite{Boutilier04:CP}. Different agents may have
arbitrarily different CP-net structure. We also assume that agents' preferences are {\em lexicographic},
meaning that agents have arbitrary importance orders over item types.


We propose
the following natural extension of TTC, which we call {\em Multi-type
  TTC (MTTC)} for multi-type housing markets. MTTC builds a
directed bipartite graph in each round, where agents and items are two
separated groups of vertices. There is an edge $(a,b_i)$, where $a$ is
an agent and $b_i$ is agent $b$'s initial type-$i$ endowment, if (1)
the most important type for agent $a$, from which she has not obtained an
item, is type $i$, and (2) $b_i$ is agent $a$'s top-ranked type-$i$
item among remaining type-$i$ items. For any type $i$ and any agent $b$, there
is an edge $(b_i,b)$. Then, MTTC finds and implements all cycles as
TTC does, but only removes the items in the cycle. The agents always remain
in the graph.

We note that in MTTC, a cycle may involve items of multiple types. Our
main theorem is the following:

\vspace{2mm}
\noindent {\bf Theorem~\ref{thm:mttc}.} {\em For lexicographic preferences, MTTC runs in polynomial-time and satisfies {\em strict-core-selection} (which implies {\em Pareto optimality} and {\em individual rationality}), {\em non-bossiness}, and {\em strong group strategy-proofness} when agents cannot lie about importance orders over types.}
\vspace{2mm}

We note that Theorem~\ref{thm:mttc}'s assumption on
lexicographic preferences are not only very mild given agents can have
arbitrarily different CP-net structure and importance orders, but also
an extension of naturalistic decision making structures studied in
cognitive science literature for ordering items based on multiple
criteria~\cite{Luan2014:From}. This positive result is especially
surprising when it is put in comparison with similar research in
{\em combinatorial voting}, where negative results emerge as soon as agents
do not share similar CP-net structures~\cite{Lang16:Voting}. 

We also prove an impossibility theorem to show that
Theorem~\ref{thm:mttc} cannot be improved by allowing agents
to lie about the importance order. As for computing other strict
cores, we prove that, it is coNP-complete to check whether a given
allocation is in the strict core, even for two types of items and agents
with separable preferences w.r.t.~the same importance order. This
hardness result is unexpected due to relatively small number of
bundles ($n^2$), and the simplicity of the underlying preference
structures.  Furthermore, the same problem for single-type housing
market can be easily computed in polynomial time. We also prove that
it is NP-hard to check whether the strict core is non-empty.

We expect that the availability of $\mttc$ will allow development of many applications
for multi-type housing markets, especially those discussed in the beginning of the Introduction. 

\subsection{Related Work and Discussions}

Housing markets are  closely related to {\em house allocation}, where agents do not have initial endowments, and {\em matching}, where houses also have preferences over agents. See~\cite{Sonmez11:Matching} for a recent survey. In the past decade, housing markets has been a popular topic under {\em multi-agent resource allocation}~\cite{Chevaleyre06:Issues}.

Many subsequent works in Economics and AI extend the standard, single-type housing markets. For example, agents may be indifferent between houses~\cite{Quint04:Houseswapping,Yilmaz2009:Random,Alcalde-Unzu2011:Exchange,Jaramillo12:Difference,Aziz12:Housing,Plaxton2013:A-Simple,Saban2013:House}. Agents may desire multiple houses~\cite{Papai2007:Exchange,Todo2014:Strategyproof,Sonoda2014:Two-Case,Fujita2015:A-Complexity,Sun2015:Exchange}. Some agents may not have initial endowments~\cite{Abdulkadiroglu99:House,Chen2002:Improving,Sonmez10:House}.

In our setting, agent's preferences are represented by the celebrated CP-nets~\cite{Boutilier04:CP}. A CP-net allows agents to specify preferences over items within each type given other items allocated to her. In general, a CP-net represents a {\em partial order} over bundles of items. CP-nets have been heavily used in combinatorial voting, see for example~\cite{Rossi04:mCP,Lang07:Vote,Lang16:Voting}, but negative results often emerge as soon as agents' CP-nets do not line up. {\em Lexicographic orders}~\cite{Booth10:Learning} are special linear orders that extend some CP-nets, where agents can specify importance orders over types. 

Previously, \citeay{Bouveret09:Conditional} proposed a CP-net-like language and~\citeay{Fujita2015:A-Complexity} used lexicographic orders to model agents' preferences over bundles of items of a single type. \citeay{Monte2015:Centralized} characterized strategy-proof mechanisms for multi-type house allocation, when agents' preferences include lexicographic preferences. To the best of our knowledge, our paper is the first time that CP-nets and lexicographic preferences are investigated in multi-type housing markets. 


\citeay{Konishi01:On-the-ShapleyScarf} assumed that agents have separable and additive preferences. This assumption is in general incomparable to our assumptions on agents' preferences, because preferences represented by CP-nets are generally not separable, except when the CP-net has no edges. Therefore, our positive results are {\em not} obtained by putting further restrictions on agents' preferences. 

There has been very little work on multi-type housing markets after~\citeauthor{Konishi01:On-the-ShapleyScarf}'s negative results~\cite{Wako2005:Coalition-proof,Klaus2008:The-coordinate-wise}. These papers focus on a different solution concept called {\em coordinate-wise core rule}, which are composed of type-wise strict-core allocations. We show that this naturally corresponds to the output of MTTC when agents' have separable and lexicographic preferences with a common importance order.

\section{Preliminaries}
We consider a market consisting of a set $N=\{1,...,n\}$ of {\em agents} with $p \ge 2$ {\em types} of indivisible items. For any $i\le p$, there are $n$ items of type $i$, denoted by $T_i=\{i_1,\ldots,i_n\}$. For each item $o$, $\type(o)$ is the type of $o$, that is, $o\in T_{\type(o)}$. Each agent $j \in N$ initially owns exactly one item of each type, and her endowment is denoted by a $p$-vector $O(j)$. W.l.o.g.~in this paper we let $O(j) = (j_1,\dots,j_p)$. Let $\mathcal{T} = T_1 \times \dots \times T_p$ be the set of all {\em bundles}, each of which is represented by a $p$-vector. We will often use vectors such as $\vec d$ and $\vec e$ to represent bundles, and for any $i\le p$, let $[\vec d]_i$ denote the type-$i$ item in $\vec d$. A {\em multi-type housing market} $M$ is given by the tuple $(N,\{T_1,\dots,T_p\},O)$. 

Each agent desires to consume exactly one item of each type, and her preferences are represented by a linear order over $\mathcal{T}$. A {\em preference profile} $P=(R_1,\dots,R_n)$ is a collection  of agents' preferences. In any multi-type housing market $M$, an allocation $A$ is a mapping from $N$ to $\mt$ such that for any $j\le n$, $A(j)$ is the bundle allocated to $j$. Since no item is allocated twice, we have that for any $j\ne j'$ and any $i\le p$, $[A(j)]_i\ne [A(j')]_i$. Given a market $M$, a {\em mechanism} $f$ is a function that maps agents' profile $P$ to an allocation in $M$.

\subsection{Axiomatic Properties}

A {mechanism} $f$ satisfies {\em individually rationality} if for any profile $P$, no agent prefers her initial endowment to her allocation by $f$. 
$f$ satisfies {\em Pareto optimality} if for any profile $P$, there does not exist an allocation $A$ such that (1) every agent weakly prefers her allocation in $A$ to her allocation in $f(P)$, and (2) some agent strictly prefers her allocation in $A$ to that in $f(P)$. 
$f$ is {\em non-bossy} if for any profile $P$, no agent can change any other agent's allocation without changing her own by reporting differently. 
$f$ is {\em strategy-proof} if for each agent, falsely reporting her preferences is not beneficial. 
A mechanism satisfies {\em strong group strategy-proofness} if there is no group of agents $S$ who can falsely report their preferences so that (1) every agent in $S$ gets a weakly preferred bundle, and (2) at least one agent in $S$ gets a strictly preferred bundle. 


An allocation $A$ is said to be {\em weakly blocked} by a {\em coalition} $S \subseteq N$, if the agents in $S$ can find an allocation $B$ of their initial endowments so that each agent weakly prefers allocation in $B$ to that in $A$, and some agent is strictly better off in $B$ than in $A$. 
 The {\em strict core} of a market is the set of all allocations that are not weakly blocked by any coalition. 
A mechanism $f$ is {\em strict-core-selecting}, if for any profile $P$, $f(P)$ is always in the strict core.

%
%
%

\subsection{CP-nets and Lexicographic Preferences}

A (directed) CP-net $\mn$ over $\mathcal{T}$ is defined by \begin{enumerate*}[label = (\roman*)] \item a directed graph $G=(\{T_1,...,T_p\},E)$, called the {\em dependency graph}, and \item for each $T_i$, there is a conditional preference table $\cpt_i$ that contains a linear order $\succ^i_{\vec u}$ over $T_i$ for each valuation $\vec{u}$ of the parents of $T_i$ in $G$, denoted $\pa(T_i)$.\end{enumerate*} Each CPT-entry $\succ^i_{\vec u}$ carries the following meaning: my preferences over type $i$ is $\succ^i_{\vec u}$ given that I get items $\vec u$, and these preferences are independent of other items I get. An agents' preferences are {\em separable} if there are no edges in the dependency graph.

Each CP-net $\mn$ represents a partial order $\succ_\mn$, which is the transitive closure of preference relations represented by all CPT entries, which are $\{(a_i,\vec{u},\vec{z}) \succ_\mn (b_i,\vec{u},\vec{z}): i \le p; a_i,b_i \in T_i; \vec{u} \in \mathcal{T}_{Pa(T_i)}; \vec{z} \in \mathcal{T}_{-(Pa(T_i) \cup \{T_i\})}\}$.

For example, Figure~\ref{fig:cp_eg} illustrates a separable CP-net. There are two types: houses (H) and cars (C), with two items per type. 
Since the preferences are separable, there is no edge in the dependency graph (in the left of the figure). The CPTs are shown in the middle of the Figure~\ref{fig:cp_eg}, and the partial order represented by the CP-net is shown in the right. 

\begin{figure}[htpb!]
\centering
\resizebox{\linewidth}{!}{%
\begin{tikzpicture}[remember picture]
\node (A)[]{${H}$};
\node (B)[below=1cm of A]{${C}$};
\node (At)[right=0.5cm of A]{
\begin{tabular}{|c|}\hline
Pref.~over $H$ \\ \hline
$1_H \succ 2_H$ \\ \hline
\end{tabular}
};

\node (Bt)[shape=rectangle,right=0.5cm of B] {
\begin{tabular}{|c|}\hline
Pref.~over $C$ \\ \hline
$1_C \succ 2_C$ \\ \hline
\end{tabular}
};

\node (C00)[right=1.2cm of At]{$(1_H,1_C)$};
\node (C01)[right=1cm of C00]{$(1_H,2_C)$};
\node (C10)[below=1cm of C00]{$(2_H,1_C)$};
\node (C11)[right=1cm of C10]{$(2_H,2_C)$};
\draw[->] (C00) -- (C01);
\draw[->] (C00) -- (C10);
\draw[->] (C10) -- (C11);
\draw[->] (C01) -- (C11);
\end{tikzpicture}
}
\caption{\label{fig:cp_eg} A CP-net. 
}
\end{figure}
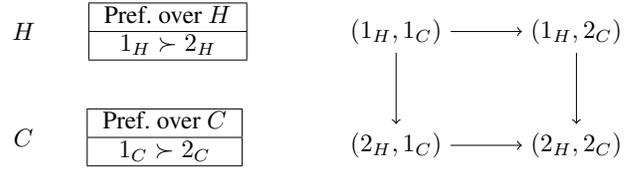

Let $\mathcal{O} = [T_1 \rhd \dots \rhd T_p]$ be a linear order over the types. A CP-net is $\mo$-legal, if there is no edge $(T_k,T_l)$ with $k>l$ in its dependency graph. A {\em lexicographic extension} of an $\mathcal{O}$-legal CP-net $\mn$ is a linear order $V$ over $\mathcal{T}$, where for any $i \le p$, any $\vec{x} \in T_1 \times \dots \times T_{i-1}$, any $a_i,b_i \in T_i$, and any $\vec{y},\vec{z} \in T_{i+1} \times \dots \times T_p$, if $a_i \succ_{\vec{x}}^i b_i$ in $\mn$, then $(\vec{x},a_i,\vec{y}) \succ_V(\vec{x},b_i,\vec{z})$. In other words, the agent believes that type $T_1$ is most important type to her, $T_2$ is the second most important type, etc. In a lexicographic extension, $\mo$ is called the {\em importance order}.

In this paper an agent's preferences are {\em lexicographic}, which means that each agent's ranking is a lexicographic extension of a CP-net. We note that the CP-net does not need to be separable and the importance order can be different. 


\begin{myeg} 
Suppose the agent's preferences is lexicographic w.r.t.~the separable CP-net in Figure~\ref{fig:cp_eg} and the importance order $H\rhd C$, then her preferences are $(1_H,1_C)\succ (1_H,2_C)\succ (2_H,1_C)\succ (2_H,2_C)$. If her importance order is $C\rhd H$, then her preferences are $(1_H,1_C)\succ (2_H,1_C)\succ (1_H,2_C)\succ (2_H,2_C)$.
\end{myeg}

\begin{figure*}[]
\centering
\begin{tabular}{c}
\begin{tabular}{ccc}
\begin{tabular}{|c|c|} \hline
Agent $1$ & $H \rhd C$ \\ \hline
$H$ prefs. & $2_H \succ $ others \\ \hline
$C$ prefs. & $3_C \succ 1_C \succ 2_C$ \\ \hline
\end{tabular} & 
\begin{tabular}{|c|c|} \hline
Agent $2$ & $H \rhd C$ \\ \hline
$H$ prefs. & $3_H \succ $ others \\ \hline
$C$ prefs. & $2_C\succ $ others \\ \hline
\end{tabular} &
\begin{tabular}{|c|c|} \hline
Agent $3$ & $H \rhd C$ \\ \hline
$H$ prefs. & $1_H \succ $ others \\ \hline
$C$ prefs. & $3_C \succ 1_C \succ 2_C$ \\ \hline
\end{tabular}
\end{tabular} \\ \\
\includegraphics[trim=0cm 9.9cm 0cm 0cm, clip=true, width=1\textwidth]{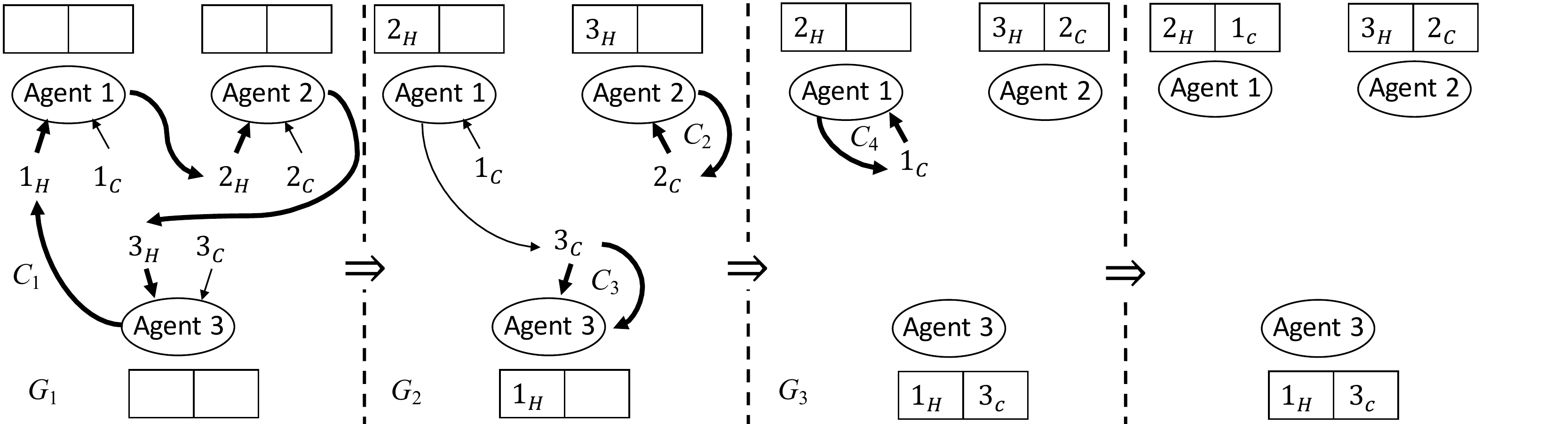} \\
\end{tabular}
\caption{\label{fig:mttc} A lexicographic and separable profile $P$ and an execution of $\mttc$ on $P$.}
\end{figure*}

\section{The Multi-Type TTC Mechanism}

We propose the {\em multi-type TTC (MTTC)} mechanism as Algorithm~\ref{alg:mttc}. $\mttc$ assumes that agents' preferences are lexicographic (w.r.t.~possibly different importance orders and possibly different CP-net structures).

\begin{algorithm}[H]
\begin{algorithmic}[1]
\State {\bf Input:} A multi-type housing market $M$ and a profile $P$ of lexicographic preferences.
\State $t \gets 1$. Let  $L \gets \cup_{i\le p} T_i$ be the set of unallocated items. Let $A$ be the empty assignment. For each $j\le n$, let $i_j^*$ is agent $j$'s most desirable type.
\While{$L \ne \emptyset$}
\parState{Build a directed graph $G_t=(N \cup L,E)$. For every $j \in N, j_i \in L$, add edge $(j_i,j)$ to $E$. For every $j \in N$, add edge $(j,\hat{j}_{i_j^*})$ to her most preferred item in $L$ of type $i_j^*$, to $E$.}
\parState{{\bf Implement cycles in $G_t$.} For each cycle $C$, for every $(j,\hat{j}_{i_j^*}) \in C$, assign $[A(j)]_{i_j^*} = [O(\hat{j})]_{i_j^*}$.}
\State Remove assigned items from $L$.
\parState{For any agent $j$ who is assigned an item, set $i_j^*$ to be the next type according to $j$'s importance order.}
\State $t \gets t+1$.
\EndWhile
\State {\bf Output:} The allocation $A$.
\end{algorithmic}
\caption{\label{alg:mttc} MTTC}
\end{algorithm}

\begin{myeg}\label{eg:mttc}
Consider the market with $3$ agents and $2$ types: Houses ($H$) and Cars ($C$) with items $\{1_H,2_H,3_H\}$ and $\{1_C,2_C,3_C\}$, respectively. Let the initial endowments of each agent $j$ be $(j_H,j_C)$. Figure~\ref{fig:mttc} shows agents' lexicographic preferences and the execution of $\mttc$. In particular, Figure~\ref{fig:mttc} shows the graphs $G_1,G_2,G_3$ constructed in each round, and all the four cycles implemented in $\mttc$. The output is the allocation $A$ such that  $A(1) = (2_H,1_C), A(2) = (3_H,2_C), A(3) = (1_H,3_C)$. 
\end{myeg}
It is not hard to see that when all agents have the same importance order, the output of $\mttc$ is the same as outcome of the following $p$-step process. W.l.o.g.~let the importance order is $T_1\rhd T_2\rhd\cdots\rhd T_p$. For each step $i$, we ask the agents to report their preferences over $T_i$, conditioned on items they got in previous steps; then we use TTC to allocate the item, and move to the next step. In particular, when all agents' preference are separable, $\mttc$ coincides with the {\em coordinate-wise core rule}~\cite{Konishi01:On-the-ShapleyScarf}.
\begin{mythm}\label{thm:mttc} When agents'
  preferences are lexicographic, MTTC runs in polynomial-time and satisfies {\em strict-core-selecting} (which implies {\em Pareto optimality} and {\em individual rationality}) and {\em non-bossiness}, and is {\em group strategy-proof} when agents cannot lie about importance orders over types.
\end{mythm}
\begin{proof} The proof is done in the following steps.

\noindent{\bf Step 1: $\mttc$ outputs an allocation in polynomial time.}
At any round $t$, if $L$ is non empty, there is a cycle in $G_t$ since every node has either: an indegree of $1$ and an outdegree of $1$, or no incoming and no outgoing edges.

At every round the algorithm permanently assigns exactly as many items as there are agents in the cycles, one to each agent. There are exactly as many items of each type as there are agents who have not yet been assigned an item of the type.

At every round, if an agent $j$ is involved in a cycle, she receives exactly one item, the item she has an outgoing edge to. An item once assigned is removed and is never re-assigned. Once an agent receives an item of a given type, she may never receive another item of the same type. The algorithm does not terminate if any items remain unassigned.

At every round of the algorithm, there is at least one agent who is permanently assigned an item of some type. Thus, there are at most $np$ rounds. At each round we must compute the pointing graph $G$. For every agent, there is a single most desired type. It takes $O(n)$ time to find the most preferred remaining item. The algorithm runs in $O(n^2p)$ time.

\vspace{2mm}
\noindent{\bf Step 2: $\mttc$ satisfies strict-core-selecting.}
For the sake of contradiction, let the allocation $A$ produced by the algorithm does not belong to the strict core. Then, there exists a coalition $S$, and an allocation $B$ on $S$ that blocks $A$. Let $t^*$ be the first round in which the algorithP's assignments restricted to $S$ at the end of the round differs from the assignments in $B$. Then, there is some $j \in S$, a type $i$ such that at least one of the following cases hold: 
\begin{enumerate}[label=(\arabic*),topsep=\parskip]\item $j$ is assigned a different item of type $i$ in $B$ than in $A$, or \item the item $j_i$ initially endowed to an agent in $S$ is assigned to a different agent in $B$ than in $A$.\end{enumerate}

Suppose Case (1) holds. There exists some agent $j \in S$, a type $i$ such that $j^*_i = [B(j)]_i \neq [A(j)]_i = \hat{j}_i$. Note that $t^*$ is the first iteration where the allocations differ, and the algorithm assigns items of agents' most desired type. Then, we must have that $j^*_i \succ \hat{j}_i$ since for every type that takes precedence over $i$, the allocation in $B$ is the same as the allocation in $A$. Now, if $j^*_i$ was available at round $t^*$, then we must have that the only outgoing edge from $j$ points at the owner $j^*_i$ and not at $\hat{j}_i$ at round $t$. This is a contradiction to the assumption that $A$ is an output of the algorithm. If $j^*_i$ was unavailable at round $t^*$, then it must already have been assigned to another agent in a strictly earlier round $t' < t^*$ which is a contradiction to our assumption that $t^*$ is the first round where assignments restricted to $S$ differ.

Now, consider Case (2). We will show that it reduces to Case (1). Suppose different agents $\hat{j}$ and $j^*$ receive item $j_i$ in $A$ and $B$ respectively i.e. $[A(\hat{j})]_i = j_i = [B(j^*)]_i, \hat{j} \neq j^*$.
Let $C$ be the cycle that is implemented at round $t^*$ of the algorithm. W.L.O.G. let $C$
consist of edges $(1,2_{i_1}),\dots,(k,1_{i_k}),(1_{i_k},1)$, involving agents $1,\dots,k$ whose most desirable types are $i_1,\dots,i_k$ at round $t^*$. 

Let agent $1 \in S$, $1_{i_k}$ is an item owned by an agent in $S$, and the agents who receive $1_{i_k}$ differs in $A,B$. Now, if $k \in S$, we must have that $[B(k)]_{i_k} \neq [A(k)]_{i_k} = 1_{i_k}$ since $1_{i_k}$ is being assigned to a different agent in $B$ and this reduces to Case (1). If $k \notin S$, there must be some consecutive pair of nodes $j,j+1_{i_j}$ in $C$ such that $2 \le j \le k, j+1 \notin S$. Then, we must have that $[B(j)]_{i_j} \neq [A(j)]_{i_j}$ since $j$ can only be assigned items initially endowed to agents in $S$ in the allocation $B$. Again, this reduces to Case (1).

\vspace{2mm}
\noindent{\bf Step 3: Defining $\mttcs$, the single-cycle-elimination variant of $\mttc$.}
To prove the non-bossiness, monotonicity, and strong group-strategy-proofness, we will consider a class of algorithms called {\em MTTC*}, which is similar to MTTC except that in each round a {\em single} trading cycle is implemented. There are multiple MTTC* algorithms, each corresponds to a different order of implementing cycles. Just by definition we do not know yet whether different MTTC* algorithms correspond to different allocations. Later we will show that, indeed different MTTC* algorithms on the same preferences must output the same allocation. In fact, we will prove a stronger lemma (Lemma~\ref{lem:mttcs=mttc}) that characterizes the executions of all MTTC* algorithms. 

\begin{mydef} Given a housing market $M$ and any profile $P$, let $\mttcs(P)$ denote the set of algorithms, each of which is a modification of \mttc{} (Algorithm~\ref{alg:mttc}), where instead of implementing all cycles in each round, the algorithm implements exactly one available cycle in each round. 
\end{mydef}
We note that $\mttcs(P)$ depends on $P$ because for different profiles the cycles in each round might be different. For each $\ma\in\mttcs(P)$, we let $\order(\ma)$ denote the linear order over the cycles that $\ma$ implements. That is, if $\order(\ma)=C_1\succ \cdots\succ C_k$, then it means that for any $t\le k$, $C_t$ is the cycle implemented by $\ma$ in round $t$.

\begin{myeg}\label{ex:mttcstar} Let $M,P$ be the same as in Figure~\ref{fig:mttc}. Let $C_1,C_2,C_3,C_4$ be the same cycles in Figure~\ref{fig:mttc}. Let $\ma$ be the \mttcs{} algorithm such that for any $t\le 4$, in the $t$-th round $C_i$ is implemented. Then, $\order(\ma)=C_1\succ C_2\succ C_3\succ C_4$.
\end{myeg}

\begin{mydef}\label{dfn:pom}
For any multi-type housing market $M$, let $\cyc(P)$ denote the set of cycles implemented in the execution of $\mttc$ on $P$. We define a partial order $\po(P)$ over $\cyc(P)$ as follows. For every pair of cycles $C_k,C_l$, $C_k \succ C_l$ in $\po(P)$ if one of the following two conditions hold:
\begin{itemize}
\item[1.] There is an agent who gets an item of a more important type in $C_k$ than in $C_l$. That is, $\exists j \in C_k \cap C_l$ such that $\type(C_k(j)) \rhd_j \type(C_l(j))$.
\item[2.]  There is an agent in $C_l$ who prefers an item in $C_k$ over the item she is pointing to in $C_l$ of the same type, conditioned on the item she got in previous rounds. That is, there exists $j \in C_l$ and $j'\in C_k$ such that, $\type(C_k(j'))=\type(C_l(j))$ and $C_k(j')\succ_{j}C_l(j)$ conditioned on the items $j$ got in previous rounds.
\end{itemize}
Then, $\po(P)$ is the transitive closure of the two classes of binary relations mentioned above. Let $\ext(P)$  denote the set of linear orders that extend $\po(P)$. 
\end{mydef}

\begin{myeg}\label{ex:pom} Continuing Example~\ref{ex:mttcstar}, $\po(P)$ is illustrated in Figure~\ref{fig:pom}. $\po(P)=\{C_1 \succ C_2, C_1 \succ C_3, C_1 \succ C_4, C_3 \succ C_4\}$. For all $2\le i\le 4$, $C_1 \succ C_i$ because houses are more important than cars to all agents. $C_3 \succ C_4$ since agent $1$ has a more preferred car in $C_3$ than in $C_4$.
\end{myeg}

\begin{figure}[H]
\centering
\includegraphics[trim=0cm 5.7cm 28.7cm 10.8cm, clip=true, width=0.2\textwidth]{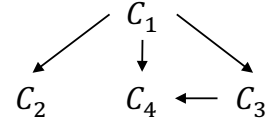}
\caption{$\po(P)$ in Example~\ref{ex:pom}.\label{fig:pom}}
\end{figure}

We are now ready to present the key lemma that establishes the equivalence between  $\mttcs(P)$ and $\ext(P)$. 
\begin{mylemma}\label{lem:mttcs=mttc} For any multi-type housing market $M$ and any profile $P$, we have $\order(\mttcs(P))=\ext(P)$.
\end{mylemma}
\begin{proof} We first prove the $\supseteq$ direction. For any $W\in \ext(P)$, we will prove that there exists an $\mttcs(P)$ algorithm $\ma$ that implements cycles exactly as in $W$. W.l.o.g.~let $W=C_1\succ\cdots\succ C_k$. Suppose for the sake of contradiction, $W$ is not in $\order(\mttcs(P))$. Let $1\le h\le k$ denote the smallest number such that there exists $\ma\in\mttcs(P)$ whose top $h-1$ cycles in $\order(\ma)$ are exactly $C_1\succ \cdots \succ C_{h-1}$ but whose $h$-th cycle in $\order(\ma)$ is not $C_{h}$; let $\ma$ denote this $\mttcs$ algorithm. Let $G_{h}$ denote the graph at the beginning of round $h$ of algorithm $\ma$. 

We now prove that $C_{h}$ must be in $G_{h}$ by showing that each edge $(j, o_i)$ in $C_h$ must be in $G_{h}$. Suppose for the sake of contradiction that this is not true and $(j, o_i)$ is in $C_{h}$ but not in $G_{h}$. There are three cases: (1) agent $j$ has not obtained an item in a more important type than $i$ yet; (2) agent $j$ is pointing to a more preferable item in type $i$ than $o_i$; and (3) agent $j$ has obtained an item in type $i$ in a previous round. Case (1) and (2) correspond to the two cases in the definition of $\po(P)$ (Definition~\ref{dfn:pom}), respectively. Therefore, neither can hold because $W$ is an extension of $\po(P)$, which means that the more preferable items must have been allocated in the first $h-1$ rounds of $\ma$. The third case cannot be true either, because otherwise it means that there exists $C_l$ with $l\le h-1$ where $j$ points to a different item in type $i$ than $o_i$. This means that agent $j$ gets two items in type $i$ in $\mttc(P)$, which is a contradiction. Therefore, $C_{h}$ must be in $G_{h}$. 

However, this contradicts the minimality of $h$. Therefore, $W\in \order(\mttcs(P))$. This proves the $\supseteq$ direction.

We now prove the $\subseteq$ direction. Suppose for the sake of contradiction there exists $\ma\in\mttcs(P)$ such that $\order(\ma)\not\in \ext(P)$. W.l.o.g.~let $\order(\ma)=[C_1\succ \cdots \succ C_{h-1}\succ C_{h}^*\succ \cdots]$, such that there exits no $L\in\ext(P)$ that agrees with $\order(\ma)$ with the order of the top $h-1$ elements (cycles), but there does not exist $L^*\in\ext(P)$ that agrees with $\order(\ma)$ with the order of the top $h$ cycles. W.l.o.g.~let $L=[C_1\succ \cdots \succ C_k]$, where $C_{h}\ne C_{h}^*$. By the $\supseteq$ direction, there exists $\ma_L\in\mttcs(P)$ such that $\order(\ma_L)=L$. 

We claim that $C_{h}^*\in\{C_{h+1},\ldots,C_k\}$. Let $G_{h}$ denote the graph at the beginning of round $h$ in $\ma$. $G_{h}$ must also be the graph at the beginning of round $h$ in $\ma_L$, because the first $h-1$ cycles in $\order(\ma)$ and those in $\order(\ma_L)$ are exactly the same. It follows that $C_{h}^*$ is in $G_{h}$. Therefore, in one of the following rounds in $\ma_L$, $C_h^*$ must be implement, otherwise the items involved in it will never be allocated.

Let $C_h^*=C_l$ for some $h+1\le l\le k$. There must exist $h+1\le l'\le k$  such that $C_{l'}\succ C_l$ in $\po(P)$, otherwise there is an extension $\po(P)$ where the top $h$ cycles are $C_1\succ\cdots\succ C_h\succ C_l$, which contradicts the assumption that no order in $\ext(P)$ agrees with $\order(\ma)$ on the top $h$ cycles. However, by the definition of $\po(P)$ and by the fact that $C_{l'}$ has not been implemented in the first $h-1$ round of $\ma$ (as well as $\ma_L$), it is impossible that $C_l$ is in $G_h$. This contradicts the assumption that $C_h^*=C_l$ is implemented in round $h$ by $\ma$. This proves the $\subseteq$ direction.
\end{proof}


It follows directly from Lemma~\ref{lem:mttcs=mttc} that the outcomes of all $\mttcs(P)$ are the same, which is $\mttc(P)$.
The next lemma states that under $\mttc$, no single agent cannot change the output of $\mttc$ by claiming that the bundle she got in the original market via $\mttc$ is her top choice. This property is similar to {\em monotonicity} for voting rule.

\begin{figure*}
\centering
\begin{tabular}{ccc}
\begin{tabular}{|c|c|} \hline
Agent $1$ & $H \rhd C$ \\ \hline
$H$ prefs. & $2_H \succ 1_H \succ 3_H$ \\ \hline
$C$ prefs. & $1_C \succ 3_C \succ 2_C$ \\ \hline
\end{tabular} & 
\begin{tabular}{|c|c|} \hline
Agent $2$ & $C \rhd H$ \\ \hline
$H$ prefs. & $3_H \succ 2_1\succ 1_1$ \\ \hline
$C$ prefs. & $2_H \succ $ others \\ \hline
\end{tabular} &
\begin{tabular}{|c|c|} \hline
Agent $3$ & $H \rhd C$ \\ \hline
$H$ prefs. & $1_H \succ 3_H \succ 2_H$ \\ \hline
$C$ prefs. & $1_C \succ 3_C \succ 2_C$ \\ \hline
\end{tabular}
\end{tabular}
\caption{\label{fig:eg_impossibility} An example of $3$ agents' separable lexicographic preferences over $2$ categories, Houses ($H$) and Cars ($C$).}
\end{figure*}

\begin{mylemma}\label{lem:nb}
For any multi-type housing market $M$, any profile $P$, and any agent $j$, let $P'$ be any market where agent $j$ changes her top bundle to be $\mttc(P)(j)$ and other agents' preferences are the same, then $\mttc(P')=\mttc(P)$.
\end{mylemma}
\begin{proof}
W.l.o.g.~let $j=1$ and $\mc_1=\{C_{k_1},C_{k_2},\ldots,C_{k_p}\}$ be the set of cycles that involve agent $j$ in $\cyc(P)$. Let $L\in \ext(P)$. 
By Lemma~\ref{lem:mttcs=mttc}, there exists $\ma\in \mttcs(P)$ such that $\order(\ma)=L$.

We next prove that $\ma$ can be applied to $P'$ and the output is the same as $\ma(P)$. To see this, we will examine two parallel runs of the execution of $\ma$ with input $P$ and with input $P'$, respectively. For any $t\ge 0$, let $G_t$ and $G_t'$ denote the graphs of $\mttcs$ in the beginning of round $t$ w.r.t.~$P$ and $P'$, respectively. We claim that for all $t\le k$, $G_t=G_t'$ by induction. The base case of $t=0$ is straightforward. Suppose the claim is truth for all $t\le h-1$. Then if the $h$-th cycle in $L$ is not in $\mc_1$, then $G_h=G_h'$ because all agents except agent $1$ have the same preferences in $P$ and in $P'$. If the $h$-th cycle in $L$ is in $\mc_1$, then suppose agent $1$ points to an item $o_i\in T_i$ in $G_t$. This means that $[\mttc(P)(1)]_i=o_i$. Therefore, agent $1$ also points to $o_i$ in $G_t'$. This proves that $G_t=G_t'$ for all $t\le k$. 

Therefore, $\ma$ can be successfully applied to $P'$ and the allocation is the same as $\mttc(P)$. By Lemma~\ref{lem:mttcs=mttc}, $\mttc(P')$ is the same as the output of $\ma$ on $P'$, which proves the lemma.
\end{proof}

\vspace{2mm}
\noindent{\bf Step 4: $\mttc$ satisfies non-bossiness.} Suppose for the sake of contradiction $\mttc$ does not satisfy non-bossiness. W.l.o.g.~let $M$ and $P'$ denote two markets where only agent $1$'s preferences are different, $\mttc(P)(1)=\mttc(P')(1)$, yet $\mttc(P)\ne \mttc(P')$. Let $\hat M$ denote the market obtained from $M$ by letting agent $1$'s top-ranked bundle to be $\mttc(P)(1)$. By Lemma~\ref{lem:nb}, $\mttc(\hat M)=\mttc(P)$ and $\mttc(\hat M)=\mttc(P')$, which is a contradiction.

\vspace{2mm}
\noindent{\bf Step 5: 
$\mttc$ is strong group-strategyproof when the agents cannot lie about the importance orders.}
Suppose for the sake of contradiction that the proposition is not true in a multi-type housing market $M$. Let $P$ denote the truthful profile, $S\subseteq N$ denote the group of strategic agents, and $P'$ denote the untruthful profile where preferences of all truthful agents are the same as in $P$. Let $\hat P$ denote the profile obtained from $P'$ by letting the top-ranked bundle of all agents $j\in S$ be $\mttc(P')(j)$. By sequentially applying Lemma~\ref{lem:mttcs=mttc} to all agents in $S$, we have that $\mttc(\hat P)=\mttc(P')$. In particular, all agents in $S$ get the same bundles in $\mttc(\hat P)$ as those in $\mttc(P')$.

We now compare side by side two parallel runs of two $\mttcs$ algorithms: $\ma\in\mttcs(P)$ and $\hat\ma\in\mttcs(\hat P)$. We will define   $\ma$ and $\hat \ma$ dynamically. Starting with $t=0$, let $G_t$ and $\hat G_t$ denote the graphs of $\mttcs$ at the beginning of round $t$ for input $P$ and input $\hat P$, respectively. If there is a common cycle $C$ in $G_t$ and $\hat G_t$, then we let both $\ma$ and $\hat \ma$ implement $C$, and move on to the next round. 

Because $\mttc(P)\ne \mttc(P')=\mttc(\hat P)$, there exists a round $t$ such that $G_t$ and $\hat G_t$ does not have a common cycle. Let $t^*$ be the earliest such round.  We note that for any $j\not\in S$, the outgoing edge of $j$ in $G_{t^*}$ and that in $\hat G_{t^*}$ must be same, because $j$ is truthful and the remaining items in $G_{t^*}$ and in $\hat G_{t^*}$ are the same. Let $C$ denote an arbitrary cycle in $G_{t^*}$. It follows that there exists $j\in S$ such that $(j,o_{i})\in C$ and $(j,s_{i^*})\in \hat G_{t^*}$, where $o_{i}\ne s_{i^*}$. Because no agent is allowed to lie about the importance order, and the allocation of all previous rounds are the same in $\ma$ and in $\hat \ma$, we must have $i=i^*$, $s_{i^*}$ is in $G_t$, and the items agent $j$ gets in all previous rounds in $\ma$ is the same as that in $\hat\ma$. Therefore, agent $j$ strictly prefers $o_i$ to $s_{i^*}$ give her allocation of more important types in previous rounds in $\ma$/$\hat\ma$. We note that because agent $j$'s top-ranked bundle is her final allocation in $\mttc(\hat P)$, agent $j$ must get $s_{i^*}$ in $\mttc(\hat P)$. Because agent $j$ gets $o_i$ in $\mttc(P)$, we have that $\mttc(P)(j)\succ_j\mttc(\hat P)$. This contradicts the assumption that none of agents in $S$ is strictly worse off in $\mttc(\hat P)$. Therefore, $\mttc$ is strong group-strategyproof when agents cannot lie about the importance orders.
\end{proof}
 
\begin{myprop}
$\mttc$ is not strategy-proof w.r.t.~only misreporting the importance order (i.e.~without misreporting local preferences over types).
\end{myprop}
\begin{proof} Consider the preferences in Figure~\ref{fig:mttc}. We recall that when agents are truthful, the output of $\mttc$ is $(2_H,1_C),(3_H,2_C),(1_H,3_C)$ to agents $1,2,3$, respectively (Example~\ref{eg:mttc}). If agent $1$ misreport the importance order as $C \rhd H$ without misreporting any preferences over types, then the output of $\mttc$ is $(2_H,3_C),(3_H,2_C),(1_H,1_C)$ to agents $1,2,3$, respectively. We note that agent $1$ prefers $(2_H,3_C)$ to $(2_H,1_C)$. This proves the proposition.
\end{proof}

\begin{myprop}
The strict core of a multi-type housing market can be multi-valued, even when agents' preferences are separable and lexicographic w.r.t.~the same order.
\end{myprop}
\begin{proof}
Consider the preferences in Figure~\ref{fig:mttc}. Let $B$ denote the allocation such that $B(1) = (2_H,3_C), B(2) = (3_H,2_C), B(3) = (1_H,1_C)$. For the sake of contradiction, suppose $S$ be a blocking coalition to $B$. We can observe that agents $1,2$ each receive their top bundles in $B$ and have no incentive to participate in a coalition. Therefore $S=N$. However, it can be verified that $B$ is Pareto optimal. Further, agent $3$ cannot benefit by not participating since $(1_H,1_C)\succ (3_H,3_C)$. This means that there is no coalition that blocks $B$, which is a contradiction.
\end{proof}

The next theorem states that, not only is $\mttc$ susceptible to misreporting importance order, but also, all mechanisms that satisfies individual rationality and Pareto optimality do. We note the $\mttc$ satisfies {\em strict-core-selecting}, which implies individual rationality and Pareto optimality.
\begin{mythm}\label{thm:impossibility} For any multi-type housing market $M$ with $n\ge 3$ and $p\ge 2$, there is no mechanism that satisfies individually rationality, Pareto optimality, and strategy-proofness, even when agents' preferences are lexicographic and separable.
\end{mythm}
\begin{proof}
For the sake of contradiction, let $f$ be any mechanism that is individually rational, Pareto-optimal and strategy-proof. 
Consider the agents' lexicographic and separable preferences $P$ as in Figure~\ref{fig:eg_impossibility}. Explicitly, $P$ is the following\\
$1$:  $(2_H,1_C) \succ (2_H,3_C) \succ (2_H,2_C) \succ {\bf (1_H,1_C)} \succ $ others.\\
$2$: $(3_H,2_2) \succ {\bf (2_H,2_C)} \succ $ others.\\
$3$: $(1_H,1_C) \succ (1_H,3_C) \succ (1_H,2_C) \succ (3_H,1_C) \succ {\bf (3_H,3_C)} \succ $ others.

The only individually rational allocations are: 
\begin{enumerate}[topsep=\parskip,leftmargin=*,labelindent=0pt,label=(\arabic*)]
\item $((2_H,3_C),(3_H,2_C),(1_H,1_C))$.
\item $((2_H,1_C),(3_H,2_C),(1_H,3_C))$.
\item $((1_H,1_C),(2_H,2_C),(3_H,3_C))$.
\end{enumerate}
Only (i) and (ii) are also Pareto-optimal.

Since $f$ is individually rational and Pareto optimal, $f(P)$ must be either $((2_H,3_C),(3_H,2_C),(1_H,1_C))$, or $((2_H,1_C),(3_H,2_C),(1_H,3_C))$. We will show that in either case there is some agent who has an incentive to misreport her preferences.

Suppose $f(P)=((2_H,3_C),(3_H,2_C),(1_H,1_C))$. Then, consider the case where agent $1$ misreports her lexicographic order as $2 \rhd 1$. Agent $1$'s preferences over bundles is $(2_H,1_C) \succ (1_H,1_C) \succ $ others. The allocation $((2_H,3_C),(3_H,2_C),(1_H,1_C))$ is not individually rational w.r.t. the profile with agent $1$'s misreported preferences. The only allocation that is both individually rational and Pareto optimal w.r.t. the misreported preferences is $((2_H,1_C),(3_H,2_C),(1_H,3_C))$. Therefore, $f$ must select this allocation. Note that agent $1$ does strictly better, $(2_H,1_C) \succ (2_H,3_C)$, when she misreports. This contradicts the assumption that $f$ is strategy-proof.

Suppose $f(P)=((2_H,1_C),(3_H,2_C),(1_H,3_C))$. Then consider the case where agent $3$  misreports her lexicographic order as $2 \rhd 1$, and her local preferences over type $1$ as $3_H \succ 1_H \succ 2_H$. Agent $3$'s preference over bundles is $(3_H,1_C) \succ (1_H,1_C) \succ (2_H,1_C) \succ (3_H,3_C) \succ $ others. The only allocation that is both individually rational and Pareto optimal w.r.t.~the profile with agent $3$'s misreported preferences is $((2_H,3_C),(3_H,2_C),(1_H,1_C))$. Therefore, $f$ must select this allocation. Note that agent $3$ gets a strictly better bundle, $(1_H,1_C) \succ_3 (1_H,3_C)$, when she misreports her preferences.This contradicts the assumption that $f$ is strategy-proof.

Therefore, such a mechanism $f$ does not exist. This proves the theorem.
\end{proof}


\section{Computing the Membership of the Strict Core}
\begin{mydef}[InStrictCore]
Given a multi-type housing market $M$, agents' preferences $P$, and an allocation $A$, we are asked whether $A$ is a strict core allocation w.r.t.~$M$.
\end{mydef}

\begin{mythm}
InStrictCore is co-NPC even when agents have separable lexicographic preferences over $p \geq 2$ types.
\end{mythm}

\begin{proof}
We start by noting that given a blocking coalition $S$ and an allocation $B$ it is easy to check whether $S$ blocks $A$ when agents have lexicographic preferences.

We show a reduction from 3-SAT. An instance of 3-SAT is given by a formula $F$ in 3-CNF, consisting of clauses $c_1,...,c_n$ involving Boolean variables $x_1,...,x_m$ that can take on values in \{0,1\}, and we are asked whether there is a valuation of the variables that satisfies $F$.

\begin{myeg}\label{eg:3sat}
$F = (x_1 \vee x_2 \vee x_3) \wedge (\bar{x}_1 \vee \bar{x}_2 \vee \bar{x}_3)$ is a formula in 3-CNF involving variables $x_1,x_2,x_3$ and clauses $c_1 = x_1 \vee x_2 \vee x_3, c_2 = \bar{x}_1 \vee \bar{x}_2 \vee \bar{x}_3$. We will use this example to illustrate the proof.
\end{myeg}

Given an arbitrary instance $I$ of 3-SAT, we construct an instance $J$ of InStrictCore where $P$ is a profile of lexicographic preferences over $p=2$ types where all agents have the same importance order $1 \rhd 2$, and an allocation $A$ as follows:

\vspace{2mm}
{{\bf Agents:}
\begin{itemize}[topsep=\parskip]
\item For every clause $c_j$, we have an agent $c_j$.
\item For every Boolean variable $x_i$, we have agents $1_i^j,0_i^j$, one for every clause $c_j$ that involves literals of the variable $x_i$.
\item For every variable $x_i$, we have additional agents $x_i$ and $\bar{x}_i$.
\item Additionally, for every variable $x_i$, we have an agent $d_i$.
\item Lastly, we add agents $e_1,e_2$.
\end{itemize}}

\vspace{2mm}
{\bf Initial Endowments:} Every agent $a$ is initially endowed with the bundle $([a]_1,[a]_2)$.

\vspace{2mm}
{\bf Preferences and Allocations: }
We will represent agents' lexicographic preferences in $P$ with importance order $1 \rhd 2$ for every agent, and their allocated items in $A$ by 2 directed graphs $G_1,G_2$, one for each type. $G_1,G_2$ have a node for each agent.

\vspace{2mm}
A $dashed$ edge $(a,b)$ in $G_k$ represents \begin{enumerate*}[label=(\roman*)]\item the preference $[b]_k \succ [a]_k$, and \item the allocation $[A(a)]_k = [b]_k$.\end{enumerate*} A $solid$ edge $(a,b)$ indicates a strict preference $[b]_k \succ [A(a)]_k$. The absence of an edge $(a,b)$ indicates $[b]_k \prec [a]_k$.

\vspace{2mm}
\noindent \textendash~The graph $G_1$ corresponding to Example~\ref{eg:3sat} is in Figure~\ref{fig:g1}. 
\begin{figure}[htp]
\centering
\resizebox{\linewidth}{!}{%
\includegraphics[trim=0cm 0cm 18.2cm 13cm, clip=true, width=0.5\textwidth]{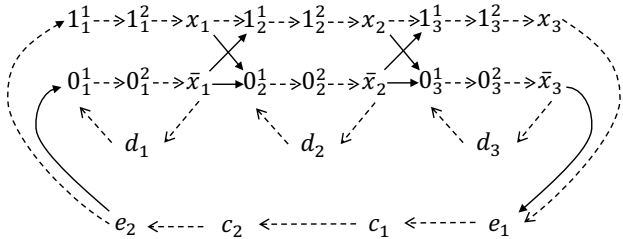}
}
\caption{\label{fig:g1} The graph $G_1$ represents agents' preferences and their allocations in $A$ over type $1$ corresponding to Example~\ref{eg:3sat}. Here, $dashed$ edges represent allocations in $A$, and preference over initial endowment, and $solid$ edges represent strict preference over allocation in $A$. For example, for agent $1_1^2$, $[A(1_1^2)]_1 = [1_2^1]_1 \succ [1_1^2]$ and $[0_2^1]_1 \succ [1_2^1]_1$.}
\end{figure}

\vspace{2mm}
Formally, we construct $G_1$ as follows: 
\begin{itemize}[topsep=\parskip]
\item For each variable $x_i$, for each $j \le k_{i-1}$, we add $dashed$ edges $(1_i^{j},1_i^{j+1})$ and $(1_i^{k_i},x_i)$.
\item For each variable $x_i$, for each $j \le k_{i-1}$, we add $dashed$ edges $(0_i^{j},0_i^{j+1})$ and $(0_i^{j},\bar{x}_i)$.
\item For $i \le m$, we add a $dashed$ edge $(x_i,1_{i+1}^1)$ and $solid$ edges $(x_i,0_{i+1}^1),(\bar{x}_i,1_{i+1}^1),(\bar{x}_i,0_{i+1}^1)$.
\item For $i \le m$, we add $dashed$ edges $(\bar{x}_i,d_i),(d_i,0_i^1)$.
\item For each of $j \le n-1$, we add $dashed$ edges $(c_j,c_{j+1})$. 
\item We add $dashed$ edges $(e_1,c_1),(c_n,e_2)$.
\item Finally, we add $dashed$ edges $(e_2,1_1^1),(x_m,e_1)$, $solid$ edges $(e_2,0_1^1),(\bar{x}_m,e_1)$.
\end{itemize}

\vspace{2mm}
\noindent \textendash~We illustrate the construction of $G_2$ in Figure~\ref{fig:g2}.
\begin{figure}[htp]
\centering
{%
\includegraphics[trim=21.5cm 0cm 0cm 10.5cm, clip=true, width=0.4\textwidth]{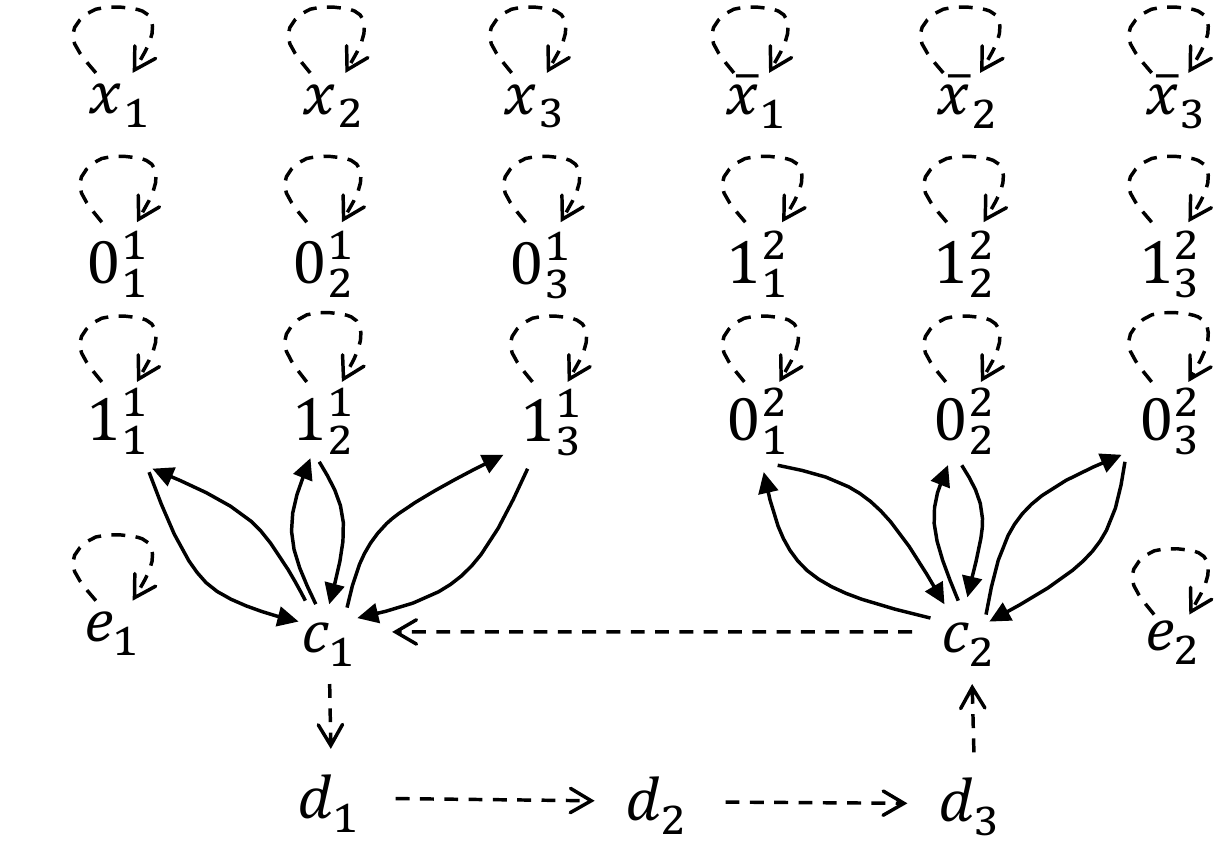}
}
\caption{\label{fig:g2} Graph $G_2$ representing agents' preferences and allocations in $A$ over type $2$ corresponding Example~\ref{eg:3sat}.}
\end{figure}

\vspace{2mm}
Formally, $G_2$ is constructed as:
\begin{itemize}[topsep=\parskip]
\item For $i \le m$, for $j \le k_i$, we add $dashed$ edges $(1_i^j,1_i^j),(0_i^j,0_i^j)$.
\item For $i \le m$, we add $dashed$ edges $(x_i,x_i),(\bar{x}_i,\bar{x}_i)$
\item We add $dashed$ edges $(c_n,c_{n-1}),\allowbreak ...,\allowbreak (c_2,c_1),\allowbreak (c_1,d_1),\allowbreak (d_1,d_2), ..., (d_{m-1},d_m), (d_m,c_n)$.
\item For every clause $c_j$, if $x_i$ is a literal in $c_j$, we add $solid$ edges $(c_j,1_i^j),(1_i^j,c_j)$ (similarly, we add $solid$ edges to and from $0_i^j$ is the literal $\bar{x}_i$ is involved in clause $c_j$).
\item We add $dashed$ edges $(e_1,e_1),(e_2,e_2)$.
\end{itemize}

\vspace{2mm}
$P$ is the profile of lexicographic preferences where agent has the importance order $1 \rhd 2$ over types and every agents' local preferences over items of type $1$ and type $2$ are represented by the graphs $G_1$ and $G_2$ respectively. This completes the construction.

\vspace{2mm}
We briefly illustrate the idea behind the rest of the proof using Example~\ref{eg:3sat}. Consider the satisfying valuation $\phi = (1,1,0)$ . It is easy to verify that there is a coalition of agents $S = \{1_1^1,1_1^2,x_1,1_2^1,1_2^2,x_2,0_3^1,0_3^2,y_2,e_1,c_1,c_2,e_2\}$. When agents in $S$ receive items according to the edges in the cycle involving all of them in $G_1$, they all receive an item of type $1$ that is weakly preferred over their item in $A$ according to their local preferences given by $G_1$. Agents $c_1$ and $c_2$ strictly improve over their allocations in $A$ according to their preferences in $G_2$ by exchanging items with agents $1_1^1$ and $0_3^2$ respectively. Now, consider the negative example of the valuation and $\psi = (1,1,1)$ and the coalition $\{1_1^1,1_1^2,x_1,1_2^1,1_2^2,x_2,1_3^1,1_3^2,x_3,e_1,c_1,c_2,e_2\}$. Note that agents $c_1$ and $c_2$ do not strictly improve over their allocations of type $1$ in $A$. They must improve over their items of type $2$. However, agent $c_2$ can only receive a strictly worse item of type $2$ compared to her allocation in $A$ from the agents in $S$.

At a high level, we use the construction in graph $G_1$ to ensure that the valuation of variables is consistent by ensuring that for any variable $x_i$, we may only exclusively have either agents $1_i^j$ or $0_i^j$ in any blocking coalition. We use $G_2$ to ensure that if there is a coalition that includes every agent $c_j$, then there is a satisfying valuation of the Boolean variables that satisfies all clauses. We now proceed with the formal proof.

\begin{myclaim} ($\Leftarrow$)
If $I$ is a Yes instance, then $J$ is a No instance.
\end{myclaim}

\begin{proof}
Let $\phi$ be a valuation that satisfies the formula $F$ in instance $I$. We will show the existence of a blocking coalition in $J$.

\vspace{2mm}
\noindent \textendash~ Consider the coalition $S$:
\begin{itemize}[topsep=\parskip]
\item $\forall i \le m$, if $\phi_i = 1$ (similarly $0$), then $\forall j \le k_i$, $1_i^j \in S$ and $x_i \in S$ (similarly $0_i^j$, $\bar{x}_i$).
\item $\forall j \le n$, $c_j \in S$.
\item $e_1,e_2 \in S$.
\end{itemize}

We will now construct an allocation $B$ on $S$ that weakly improves every agents' allocation over $A$ and is strictly improving for some agent. We will construct $B$ by identifying cycles in the graph and assigning every agent in a cycle the item of the agent she is pointing at.

\vspace{2mm}
\noindent \textendash~Consider the graph $G_1$ and the allocations corresponding to the cycle:
\begin{itemize}[topsep=\parskip]
\item Edges $(e_1,c_1),(c_1,c_2),...,(c_{n-1},c_n),(c_n,e_2)$.
\item If $\phi_1 = 1$, the edge $(e_2,1_1^1)$. (Similarly, $\phi_1=0, 0_1^1$)
\item For every $i \le m$, if $\phi_i = 1$, the edges $j \le k_i-1, (1_i^j,1_i^{j+1})$ and the edge $(1_i^{k_i},x_i)$. 
\item For every $i \le m-1$, if $\phi_i=1$ and $\phi_{i+1}=1$, the edge $(x_i,1_{i+1}^1)$.
\item If $\phi_m = 1$, the edge $(x_m,e_1)$.
\end{itemize}

\vspace{2mm}
\noindent \textendash~In $G_2$, consider the allocations corresponding to the cycles:
\begin{itemize}[topsep=\parskip]
\item For $j \le n$, let $i$ be the lowest number such that $x_i = \phi_i$ satisfies clause $c_j$, and let $\phi_i = 1$, then a cycle involving $solid$ edges $(c_j,1_i^j),(1_i^j,c_j)$. (similarly if $\phi_i = 0$, a cycle involving $c_j,0_i^j$).
\item For the remaining agents $1_i^j \in S$ (or $0_i^j \in S$), agents $x_i \in S$ (or $\bar{x}_i \in S$) and agents $e_1,e_2$, a self loop.
\end{itemize}

\vspace{2mm}
\noindent \textendash~In the allocation $B$ constructed by implementing the cycles we identified, every agent receives a weakly improving bundle over their allocated bundle in $A$ since they receive items of both types that are weakly preferred according to their local preferences. The agents $c_j$ and $1_i^j$ (or $0_i^j$) corresponding to a satisfying valuation of the variable $x_i$ involved in clause $c_j$, receive strictly improving bundles by receiving strictly preferred items of type $2$ in $B$ compared to $A$ according to their local preferences. Lastly, note that $B$ is an allocation on $S$ since none of the items allocated in $B$ were initially endowed to any agent outside $S$. This shows that $S$ is a blocking coalition and $B$ is a weakly improving allocation on $S$, and strictly improving for some agents in $S$. Therefore, $J$ is a No instance.
\end{proof}

\begin{myclaim}($\Rightarrow$)
If $J$ is a No instance, $I$ is a Yes instance.
\end{myclaim}

\begin{proof}
We are given that $J$ is a No instance. Let $S$ be a coalition that blocks $A$ w.r.t. preferences in $P$ and let $B$ be an allocation on $S$ such that \begin{enumerate*}[label=(\roman*)]\item $\forall a \in S, B(a) \succeq A(a)$, and \item $\exists a \in S, B(a) \succ A(a)$\end{enumerate*}.

\vspace{2mm}
\noindent $\bullet$~We start by proving some simple properties about the membership of agents in $S$ that are specific to our construction. Throughout, we will use the fact that $B$ must be an allocation where \begin{enumerate*}[label=(\alph*)]\item $\nexists a \in S$ such that $[B(a)]_1 \prec [A(a)]_1$, and \item $\forall a \in S$, if $[B(a)]_1 = [A(a)]_1$, then we must have that $[B(a)]_2 \succeq [A(a)]_2$\end{enumerate*}. Otherwise, if there is some agent $a^*$ for which the allocation $B$ fails to satisfy either of these conditions, then $B(a^*) \prec A(a^*)$.

\begin{mylemma}\label{lem:xij}
For any $i \le m$, any $j \le k_i-1$, \begin{enumerate*}[label=(\roman*)] \item if $1_i^j \in S$, then $[B(1_i^j)]_1 = [1_i^{j+1}]_1 = [A(1_i^j)]_1, 1_i^{j+1} \in S$, and \item if any $0_i^j \in S$, then $[B(0_i^j)]_1 = [0_i^{j+1}]_1 = [A(0_i^j)]_1, 0_i^{j+1} \in S$.\end{enumerate*}
\end{mylemma}
\begin{proof}
Let there be some agent $1_i^j \in S, i \le m, j \le k_i-1$, such that $[B(1_i^j)]_1 \neq [1_i^{j+1}]_1$. Then, according to her local preferences in $G_1$, $[B(1_i^j)]_1 \prec [A(1_i^j)]_1$ and by her importance order $1 \rhd 2$, $B(1_i^j) \prec A(1_i^j)$. This is a contradiction to our assumption that $B$ is an improving allocation for every agent in $S$.
\end{proof}

\begin{mylemma}\label{lem:xij1}
For any $i \le m$, any $2 \le j \le k_i$, \begin{enumerate*}[label=(\roman*)] \item if $1_i^j \in S$, then $[B(1_i^{j-1})]_1 = [1_i^j]_1$ and $1_i^{j-1} \in S$, and \item if $0_i^j \in S$, then $[B(0_i^{j-1})]_1 = [0_i^j]_1$ and $0_i^{j-1} \in S$.\end{enumerate*}
\end{mylemma}
\begin{proof}
By Lemma~\ref{lem:xij}, $1_i^j$ must be assigned $[1_i^{j+1}]_1$ (or if $j=k_i$, one of the items $[1_{i+1}^1]_1, [0_{i+1}^1]_1$, or $[e_2]_1$). Then, the item $[1_i^j]_1$ must be assigned to another agent in $S$. Suppose an agent $a \in S, x \neq 1_i^{j-1}$ is assigned $[1_i^j]_1$ in $B$. However, according to the preferences in $G_1$, for every agent $a$ other than $1_i^{j-1}$, the item $[1_i^j]_1$ is strictly worse than their allocated type-$1$ item in $A$. Then, $B(a) \prec A(a)$ according to $P$, a contradiction to our assumption that $B$ is weakly improving for every agent in $S$.
\end{proof}

\begin{mylemma}\label{lem:allxij}
For any $i \le m$, any $j^* \le k_i$, \begin{enumerate*}[label=(\roman*)]\item if $1_i^{j^*} \in S$, then for every $j \le k_i$, $1_i^j \in S, [B(1_i^j)]_1 = [A(1_i^j)]_1$ and $x_i \in S$, and \item if $0_i^{j^*} \in S$, then for every $j \le k_i$ $0_i^j \in S, [B(0_i^j)]_1 = [A(0_i^j)]_1$ and $\bar{x}_i \in S$.\end{enumerate*}.
\end{mylemma}
\begin{proof}
If $1_i^{j^*} \in S$, by Lemma~\ref{lem:xij}, for every $j = j^*+1,\dots,k_i$ agent $1_i^{j} \in S$ and $[B(1_i^j)]_1 = [A(1_i^j)]_1$, and by Lemma~\ref{lem:xij1}, for every $j = 1,\dots,j^*-1$, agent $1_i^j \in S$ and $[B(1_i^j)]_1 = [A(1_i^j)]_1$. If $1_i^{k_i} \in S$, $[B(1_i^{k_i})]_1 = [A(1_i^{k_i})]_1$ and $x_i \in S$.
\end{proof}

\begin{mylemma}\label{lem:xij2}
For any $i \le m$, any $j \le k_i$, if $1_i^j \in S$ and $[B(1_i^j)]_1 \preceq [A(1_i^j)]_1$, then either $[B(1_i^j)]_2 = [1_i^j]_2 = [A(1_i^j)]_2$ or $[B(1_i^j)]_2 = [c_i^j]_2$.
\end{mylemma}
\begin{proof}
Every other type-$2$ item is strictly worse than $[A(1_i^j)]_2$ according to agent $1_i^j$'s local preferences in $G_2$. If $[B(1_i^j)]_2 \notin \{[1_i^j]_2,[c_i^j]_2\}$, then $B(j) \prec A(j)$ according to her preferences in $G_1$ and $G_2$, a contradiction to our assumption that $B$ is improving for every agent in $S$.
\end{proof}

\begin{myclaim}\label{cl:cj}
Let any agent $c_{j^*} \in S$, then for every $j \le n$ \begin{enumerate*}[label=(\roman*)] \item $[B(c_j)]_1 = [A(c_j)]_1 = [c_{j+1}]_1, c_{j+1} \in S$, \item if $[B(c_j)]_2 \notin \{[a]_2: a\allowbreak~ \text{is an agent}~\allowbreak0_i^j~\allowbreak\text{or}~\allowbreak1_i^j~\allowbreak\text{corresponding\allowbreak~ to the valuation \allowbreak of a variable}~\allowbreak x_i~\allowbreak\text{that\allowbreak~ satisfies\allowbreak~ clause}~c_j\}$ then, $B(c_j) = [c_{j-1}]_2 = A(c_j)$, and \item $e_1,e_2 \in S, B(e_1) = A(e_1)$, and $B(e_2) = A(e_2)$.\end{enumerate*}.
\end{myclaim}
\begin{proof}
If $c_{j} \in S$, by the preferences in $G_1$, we must have that $[B(c_j)]_1 = [A(c_j)]_1 = [c_{j+1}]_1$ and $c_{j+1} \in S$ (and if $j=m$, $[B(c_m)]_1 = [A(c_m)]_1 = [e_2]_1$ and $e_2 \in S$) since every other item of type $1$ is strictly worse. Now, the item $[c_j]_1$ must be assigned to some agent in $S$. However, the only agent for whom the assignment of $[c_j]_1$ is weakly beneficial is the agent $c_{j-1}$ (or if $j=1$, the agent $e_1$). By induction we must have that if $c_{j^*} \in S$, every agent $c_j$ and the agents $e_1,e_2$ are in $S$ and their allocation of items of type $1$ is the same in $B$ and in $A$. This proves part (i).

To prove part (ii), let us examine $G_2$. Every other item of type $2$ is strictly worse than $[c_{j-1}]_2$ according to $c_j$'s local preferences in $G_2$. If she receives any other item of type $2$, then $B(c_j) \prec A(c_j)$, a contradiction to our assumption that $B$ is improving for every agent in $S$.

Finally, agents $e_1$ and $e_2$ must be assigned $[e_1]_2$ and $[e_2]_2$ respectively since they do not strictly improve on their item of type $1$ in $B$. This completes the proof of part (iii).
\end{proof}

\begin{mylemma}\label{lem:di}
For every $i \le m, d_i \notin S$.
\end{mylemma}

\begin{proof}
We provide a proof by contradiction that if $d_i \in S$, then $S$ must involve every agent in our construction and that there is no allocation $B$ that improves over $A$.

\begin{myclaim}\label{cl:di1}
If $d_i \in S$, $[B(d_i)]_1 = [0_i^1]_1 = [A(d_i)]_1$ and $0_i^1 \in S$.
\end{myclaim}
\begin{proof}
Let us examine $G_1$. If $d_i \in S$, we must have that $[B(d_i)]_1 = [0_i^1]_1 = [A(d_i)]_1, 0_i^1 \in S$, since every other item is strictly worse according to her local preferences over type $1$. 
\end{proof}

\begin{myclaim}\label{cl:0ijdi}
If $d_i \in S$, then for every $0_i^j$, $0_i^j \in S, [B(0_i^j)]_1 = [A(0_i^j)]_1$ and $\bar{x}_i \in S, [B(0_i^j)]_1 = [A(0_i^j)]_1 = [d_i]_1$.
\end{myclaim}
\begin{proof}
If $d_i \in S$, the claim follows from Claim~\ref{cl:di1}, Lemma~\ref{lem:allxij}, and the fact that $[\bar{x}_i]$ must be assigned to $d_i$ if $\bar{x}_i \in S$ according to the preferences of every other agent in $G_1$.
\end{proof}

\begin{myclaim}\label{cl:alldi}
If a $d_{i^*} \in S$, then \begin{enumerate*}[label=(\roman*)]\item for every $i \le n$, $d_i \in S$ and $B(d_i) = A(d_i)$, and \item for every $j \le m$, agent $c_j \in S$ and $B(c_j) = A(c_j)$.\end{enumerate*}
\end{myclaim}
\begin{proof}
Let us examine $G_2$. For every $i = 1,...,m-1$, if $d_i \in S$, we must have that $[B(d_i)]_2 = [d_{i+1}]_2 = [A(d_i)]_2, d_{i+1} \in S$. Suppose for some $i \le m-1$, $[B(d_i)]_2 \neq [d_{i+1}]_2$, then $[B(d_i)]_2 \prec [A(d_i)]_2$. For $i=m$, we must have that $[B(d_m)]_2 = [c_{n}]_2 = [A(d_m)]_2$ and that $c_n \in S$. Together with Claim~\ref{cl:di1} this proves part (i) of the claim. 

Part (ii) of the claim is proved by Claim~\ref{cl:cj} and the following argument. First we show that $c_n$ cannot enter into an exchange with an agent $1_i^n$ (or $0_i^n$) corresponding to a valuation of $x_i$ that satisfies $c_n$. 
If $c_n$ receives $[1_i^n]_2$ (or $[0_i^n]_2$), then by Lemma~\ref{lem:allxij}, Claim~\ref{cl:0ijdi} and Lemma~\ref{lem:xij2}, the agent $1_i^n$ (or $0_i^n$) must be assigned $[c_n]_2$ since she cannot be assigned a strictly improving item of type $1$ in $B$. 

Let us consider agents $c_j$ in the order $j=n-1,\dots,1$. We know that $c_j \in S$ and the item $[c_{j+1}]_2$ must already be assigned to another agent. Then, $c_j$ must exchange items with an agent $1_i^j$ (or $0_i^j$) corresponding to a valuation of the variable $x_i$ that satisfies $c_j$.
\end{proof}

\vspace{2mm}
\noindent \textendash~If $d_i \in S$, then by Claim~\ref{cl:alldi}, Claim~\ref{cl:di1}, Claim~\ref{cl:0ijdi}, Lemma~\ref{lem:allxij} and Lemma~\ref{lem:xij2}, $S$ must be the set of all agents, and that the weakly improving allocation $B$ is exactly the allocation $A$. This is a contradiction to our assumption that $B$ is weakly improving for every agent but also strictly improving for some agent.
\end{proof}

\vspace{2mm}
\noindent $\bullet$~The final key step is to establish that any blocking coalition $S$ involves all of the agents $c_j$, and for every $i \le m$, either all of $x_i^j$ or all of $\bar{x}_i^j$.

\begin{mylemma}\label{lm:satblock}
Every weakly blocking coalition $S$ consists of all the agents $c_j, j \le n$, and $\forall i \le n$, either \begin{enumerate*}[label=(\roman*)]\item all agents $x_i^j$, or \item all agents $\bar{x}_i^j$.\end{enumerate*}
\end{mylemma}

\begin{proof}
We start by showing that every blocking coalition $S$ must include at least one of the agents $1_i^j,0_i^j,x_i,\bar{x}_i$, or $c_j$ since these are the only agents with $solid$ incoming or outgoing edges and every allocation $B$ that strictly improves over $A$ must involve an assignment corresponding with one of these edges. 

\begin{myclaim}
If any $c_{j^*} \in S$ or any $0_{i^*}^{j^*} \in S$ or any $1_{i^*}^{j^*} \in S$, then $e_1,e_2 \in S$, every $c_j \in S$ and for every $i \le m$, either every agent $1_i^j \in S$ or every $0_i^j \in S$
\end{myclaim} 

\vspace{2mm}
\noindent \textendash~Suppose $c_{j^*} \in S$. Let us examine her preferences in $G_1$. By Claim~\ref{cl:cj}, for $j \le n-1$, we must have that $[B(c_j)]_1 = [c_{j+1}]_1 = [A(c_j)]_1$, and $c_{j+1} \in S$. By induction, $c_n \in S$.

\vspace{2mm}
\noindent \textendash~If $c_n \in S$, we must have that $[B(c_n)]_1 = [e_2]_1 = [A(c_n)]_1, e_2 \in S$ since every other item of type $1$ is worse according to $G_1$.

\vspace{2mm}
\noindent \textendash~If $e_2 \in S$ then either $[B(e_2)]_1 = [1_1^1]_1$, and $1_1^1 \in S$ or $[B(e_2)]_1 = [0_1^1]_1$, and $0_1^1 \in S$ since these are the only two items that are weakly preferred over $[A(e_2)]_1$.

\vspace{2mm}
\noindent \textendash~For any $i \le m$, if $1_i^{j^*} \in S$, then by Lemma~\ref{lem:allxij}, for every $j \le k_i$, $1_i^j \in S$ and $x_i \in S$. (Similarly, if for any $i \le m$, any $0_i^{j^*} \in S$, then for every $j \le k_i$, agent $0_i^j \in S$).
 
\vspace{2mm}
\noindent \textendash~For $i \le m-1$, if $x_i \in S$, then either $1_{i+1}^1 \in S$ or $0_{i+1}^1 \in S$ since in the case that $x_i$ is not assigned one of the items $[1_{i+1}^1]_1$ or $[0_{i+1}^1]_1$ in $B$, any other item of type $1$ assigned to her must be worse than $[A(1_i^{k_i})]_1$ according to the preferences in $G_1$, a contradiction to our assumption that $B$ is weakly improving for every agent in $S$. Similarly, if $\bar{x}_i \in S$, then either $1_{i+1}^1 \in S$ or $0_{i+1}^1 \in S$.

\vspace{2mm}
\noindent \textendash~ Finally, if $x_m \in S$ (or $\bar{x}_m$), then we must have that $e_1 \in S$, since $[e_1]_1$ is the only item that is weakly better than her allocation in $A$ according to $G_1$.

\vspace{2mm}
\noindent \textendash~If $e_1 \in S$, we must have that $c_1 \in S$ since $[c_1]_1$ is the only item that is weakly preferred over $[A(e_1)]_1$ according to $G_1$.

\vspace{2mm}
\noindent \textendash~Thus, if any agent $1_i^j,0_i^j,x_i$ or $\bar{x}_i$ is in $S$, then we must have $c_1 \in S$. We have already established that the existence of any $c_j \in S$ implies that for every $i \le m$, either all $1_i^j \in S$ or all $0_i^j \in S$. 

\vspace{2mm}
\noindent \textendash~Finally, we will show that there is no blocking coalition $S$ such that for some $i \le m$, any $j \le k_i, j' \le k_i$ the agents $1_i^j \in S$ and $0_i^{j'} \in S$. Suppose for the sake of contradiction that such a coalition exists. we must have by Lemma~\ref{lem:xij}, Lemma~\ref{lem:xij1} and Lemma~\ref{lem:di} that both $x_m,\bar{x}_m \in S$. However, by Lemma~\ref{lem:di}, we know that $d_m \notin S$ which implies that both $x_m$ and $\bar{x}_m$ must get $[e_1]_1$ in $B$ since every other item is strictly worse than their allocations in $A$ according to preferences in $G_1$. This is impossible since items are unique and indivisible. 
\end{proof}

\begin{mylemma}\label{lem:xij1noimprovement}
If $1_i^j \in S$, $[B(1_i^j)]_1 = [A(1_i^j)]_1 = [1_i^{j+1}]_1$ (if $j = k_i, [x_i]_1$).
\end{mylemma}
\begin{proof}
Every other item of type $1$ is strictly worse for agent $1_i^j$ according to $G_1$. Similarly, if $0_i^j \in S$, $[B(0_i^j)]_1 = [A(0_i^j)]_1 = [0_i^{j+1}]_1$ (if $j = k_i, [\bar{x}_i]_1$).
\end{proof}

\vspace{2mm}
\noindent $\bullet$~We will now prove that: {\em If there is a coalition $S$ that blocks $A$, there is a satisfying assignment to $F$.}

We now describe the construction of a valuation $\phi$ that satisfies the formula $F$ in the instance $I$ by examining $G_2$ and argue that such a valuation must exist. We start by noting that since agents $c_j$ and agents $1_i^j,0_i^j$ cannot strictly improve over their allocations of type $1$ according to preferences in $G_1$, they must weakly improve their allocation of items of type $2$. 

\vspace{2mm}
\noindent \textendash~Let us examine $G_2$. We know by Lemma~\ref{lem:di} that none of the agents $d_i$ are in $S$. We have also shown that $c_1 \in S$. Agent $c_1$ must receive an item $[1_i^1]_2$ (or $[0_i^1]$) where the corresponding valuation of the variable $x_i$ satisfies $c_1$ and the agent $1_i^1$ (or $0_i^1$) is in $S$. Otherwise, every other item of type $2$ is strictly worse than her allocated item of type $2$ in $A$ and together with the fact that $c_1$ does not receive a strictly better item of type $1$ in $B$, we have that $B(c_1) \prec A(c_1)$, a contradiction to our assumption that $B$ is weakly improving for every agent in $S$. Further, agent $1_i^1$ (or $0_i^1$) must receive the item $[c_1]_2$ since every other item is strictly worse. Set $\phi_i$ to be $1$ if $c_1$ gets $[1_i^1]_2$ and to $0$ if $c_1$ gets $[0_i^1]_2$ in $B$.

\vspace{2mm}
\noindent \textendash~Similarly, agent $c_2$ cannot assigned $[c_1]_2$ since it must be assigned to some agent $1_i^1$ (or $0_i^1$) and agent $c_2$ must exchange items with agent an agent $1_i^2$ (or $0_i^2$) that corresponds to a satisfying valuation of a variable $x_i$ that satisfies clause $c_2$. By an inductive argument, we must have that for every $j \le m$, agent $c_j$ receives the item of an agent $1_i^j$ or $0_i^j$ that is in $S$ which corresponds to a satisfying valuation of variables $x_i$ that satisfies the clause $c_j$. For each clause, set $\phi_i$ to be $1$ if $c_j$ gets $[1_i^j]_2$ and to $0$ if $c_j$ gets $[0_i^j]_2$ in $B$. Note that each clause $c_j$ is satisfied by a variable whose value is consistent with the final value according to $\phi_i$ since the value of $\phi_i$ can never change once set by Lemma~\ref{lm:satblock}. The agents involved correspond to a consistent valuation of the variables and all of the clauses must be satisfied simultaneously. 
\end{proof}

This completes the proof. Given an instance $I$ of 3-SAT, we can construct in polynomial time, a corresponding instance $J$ of InStrictCore, such that $I$ is a Yes instance iff $J$ is a No instance.
\end{proof}

\begin{mydef}[StrictCoreNonEmpty]
Given a multi-type housing market $M$, agents' preferences $P$, and an allocation $A$, we are asked whether the strict core of $M$ is non-empty.
\end{mydef}

\begin{mythm}
StrictCoreNonEmpty is NP-hard.
\end{mythm}

\begin{proof}
We will show a reduction from 3-SAT. An instance $I$ of 3-SAT is given by a formula $F$ in 3-CNF consisting of clauses $c_1,\dots,c_n$ involving Boolean variables $x_1,\dots,x_m$, and we are asked whether there is a valuation of the variables that satisfies $F$. Each clause $c_j$ is a disjunction of exactly $3$ literals $c_j = l_{j_1} \vee l_{j_2} \vee l_{j_3}$. Each literal $l_i$ is either the variable ({\em positive} literal) $x_i$ or its negation ({\em negative} literal) $\bar{x}_i$. Each of $l_{j_1},l_{j_2},l_{j_3}$ corresponds to the {\em positive} or {\em negative} literal of one of the variables $x_1,\dots,x_m$. 

\vspace{2mm}
\noindent $\bullet$~We define the following ordering $W$ on positive and negative literals of the Boolean variables where positive literals are ranked over negative literals and literals of variables with lower index are ranked over literals of higher index. Let $l_i,l_{i'}$ be any two literals. Then, $l_i \succ_W l_{i'}$ if either:
\begin{enumerate*}[label=(\arabic*)]
\item $i \le i'$ and if either \begin{enumerate*}[label=(\roman*)]\item$l_i = x_i$, or \item $l_i = \bar{x_i}$ and $l_{i'} = \bar{x}_{i'}$.\end{enumerate*}
\item $i \ge i'$, $l_i = x_i$ and $l_{i'} = \bar{x}_i'$.
\end{enumerate*}

\begin{myeg}
The order on the literals in Example~\ref{eg:3sat} is $x_1 \succ x_2 \succ x_3 \succ \bar{x}_1 \succ \bar{x}_2 \succ \bar{x}_3$.
\end{myeg}

\vspace{2mm}
\noindent $\bullet$~For every clause $c_j = l_j1 \vee l_j2 \vee l_j3$, the literals are indexed so that $l_j1 \succ_W l_j2 \succ_W l_j3$. 

\vspace{2mm}
We construct an instance $J$ of CNE as follows:

\vspace{2mm}
{\bf Agents:}
\begin{itemize}
\item For every clause $c_j$, add the agents $c_j^1,c_j^2,c_j^3$.
\item For every variable $x_i$, add agents $x_i$.
\item For every variable $x_i$, add agents $1_i^j, 0_i^j$ for every clause $c_j$ that involves variable $x_i$.
\end{itemize}

\vspace{2mm}
{\bf Preferences:}
\begin{itemize}
\item Agents $c_j^1$: $(1_{j_1}^j,c_{j+1}^1) \succ (1_{j_2}^j,c_{j+1}^1) \succ (1_{j_3}^j,c_{j+1}^1) \succ (c_j^1,c_j^3) \succ (c_j^3,c_j^3) \succ (c_j^1,c_j^2) \succ (c_j^1,c_j^1) \succ $ others. If $j = n$, replace $c_{j+1}$ with $c_1$. For $k \in \{1,2,3\}$, if $l_{j_k} = \bar{x}_{j_k}$, replace $1_{j_k}^j$ with $0_{j_k}^j$.
\item Agents $c_j^2$: $(c_j^2,c_j^3) \succ (c_j^2,c_j^1) \succ (c_j^3,c_j^3) \succ (c_j^3,c_j^1) \succ (c_j^2,c_j^2) \succ $ others.
\item Agents $c_j^3$: $(c_j^2,c_j^1) \succ (c_j^2,c_j^2) \succ (c_j^3,c_j^1) \succ (c_j^1,c_j^1) \succ (c_j^3,c_j^2) \succ (c_j^1,c_j^2) \succ (c_j^2,c_j^3) \succ (c_j^3,c_j^3) \succ (c_j^1,c_j^3) \succ $ others.
\item Agents $1_i^j$:
\begin{itemize}
\item If $l_i=x_i \in c_j, j \le k_i$, $(c_j^1,1_i^{j+1}) \succ (1_i^j,1_i^j) \succ $ others. If $j=k_i$, replace $1_i^{j+1}$ with $x_{(i+1) \mod m}$.
\item If $l_i=x_i \notin c_j, j \le k_i$, $(1_i^j,1_i^{j+1}) \succ (1_i^j,1_i^j) \succ $ others. If $j=k_i$, replace $1_i^{j+1}$ with $x_{(i+1) \mod m}$. 
\end{itemize}
\item Agents $0_i^j$:
\begin{itemize}
\item If $l_i=\bar{x}_i \in c_j, j \le k_i$, $(c_j^1,0_i^{j+1}) \succ (0_i^j,0_i^j) \succ $ others. If $j=k_i$, replace $0_i^{j+1}$ with $0_{(i+1) \mod m}$.
\item If $l_i=\bar{x}_i \notin c_j, j \le k_i$, $(0_i^j,0_i^{j+1}) \succ (0_i^j,0_i^j) \succ $ others. If $j=k_i$, replace $0_i^{j+1}$ with $0_{(i+1) \mod m}$. 
\end{itemize}
\item Agents $x_i$: $(x_i,1_i^1) \succ (x_i,0_i^1) \succ (x_i,x_i) \succ $ others.
\end{itemize}

\vspace{2mm}
{\bf Initial endowments:} Each agent $x$ is initially endowed with a bundle $O(x) = (x,x)$.\\

High level idea: 
\begin{itemize}
\item The preferences are structured so that type-$1$ items are used to track satisfaction of individual clauses, while type-$2$ are used to ensure a consistent valuation of the Boolean variables and satisfaction of all clauses.
\item Borrowing from the Example 2.2 in~\cite{Konishi01:On-the-ShapleyScarf}, if  $c_j^1,c_j^2,c_j^3$ only receive each others' items in an allocation, there is a blocking coalition. For an allocation to be individually rational and stable, agent $c_j^1$ must receive an item of type $1$ from an agent $1_i^j$ or $0_i^j$.
\end{itemize}

\begin{myeg}\label{eg:Konishi}[Example 2.2 in~\cite{Konishi01:On-the-ShapleyScarf}]
Consider the multi-type housing market with agents $N=\{1,2,3\}$, $p=2$ types and the following preferences.
\begin{itemize}[leftmargin=0pt]
\item[] Agent 1: $(1,3) \succ_1 (3,3) \succ_1 (1,2) \succ_1 $ others.
\item[] Agent 2: $(2,3) \succ_2 (2,1) \succ_2 (3,3) \succ_2 (3,1) \succ_2 (2,2) \succ_2 $ others.
\item[] Agent 3: $(2,1) \succ_3 (2,2) \succ_3 (3,1) \succ_3 (1,1) \succ_3 (3,2) \succ_3 (1,2) \succ_3 (2,3) \succ_3 (3,3) \succ_3 (1,3)$ others.
\end{itemize}

The strict core is empty for the market above.
\end{myeg}

\begin{myclaim}($\Rightarrow$)
If $I$ is a Yes instance, $J$ is a Yes instance.
\end{myclaim}
\begin{proof}
We start by define a lexicographic order $Q$ on valuations. Given two valuations $\psi,\omega \in \Phi$, let $i$ be the lowest value such that $\psi_i$ and $\omega_i$ differ. We rank $\psi$ over $\omega$ if $\psi_i=1$ ($\omega_i$ must be $0_i$). 

\begin{myeg}
In the Example~\ref{eg:3sat}, the valuation $(1,1,0)$ is ranked over $(1,0,1)$.
\end{myeg}

\vspace{2mm}
\noindent \textendash~Let $\Phi$ be the set of satisfying valuations to $F$. Let $\phi$ be the top ranked satisfying valuation according to $Q$. We will construct an allocation $A$ w.r.t. the valuations in $\phi$ that is a core allocation.\\

{\bf Allocation $A$:}
\begin{itemize}
\item $A(c_j^1) = (1_i^j,c_{j+1}^1)$ if $l_i$ is the highest ranked literal according to $W$ that satisfies $c_j$ when $x_i = \phi_i$.  If $l_i = \bar{x}_i$, replace $1_i^j$ with $0_i^j$.
\item $A(c_j^2) = (c_j^2,c_j^3)$.
\item $A(c_j^3) = (c_j^3,c_j^2)$.
\item $A(1_i^j) = $ 
\begin{itemize}
\item $(c_j^1,1_i^{j+1})$, if $l_i=x_i$ is the highest ranked literal according to $W$ that satisfies $c_j$ when $x_i=\phi_i$.
\item $(1_i^j,1_i^{j+1})$, otherwise.
\end{itemize} 
If $j=k_i$, replace $1_i^{j+1}$ with $x_{(i+1) mod m}$
\item $A(0_i^j) = $ 
\begin{itemize}
\item $(c_j^1,0_i^{j+1})$, if $l_i=\bar{x}_i$ is the highest ranked literal according to $W$ that satisfies $c_j$ when $x_i=\phi_i$.
\item $(0_i^j,0_i^{j+1})$, otherwise.
\end{itemize} 
If $j=k_i$, replace $0_i^{j+1}$ with $x_{(i+1) mod m}$
\item $A(x_i) = (x_i,x_i^1)$.
\end{itemize}

\vspace{2mm}
\noindent \textendash~Suppose, for the sake of contradiction that $S$ blocks $A$. Let $B$ be an improving allocation on $S$. We will show that $B$ corresponds to a satisfying assignment to $F$ that is ranked lower than $\phi$, and $B$ must be strictly worse for some agent in $S$.

\vspace{2mm}
\noindent \textendash~We start by showing that $S$ must involve at least one of the agents $c_j^1, 1_i^j$, or $0_i^j$. 

If either of $c_j^2$ or $c_j^3$ are in $S$ for some $j \le n$, then we must have that all the agents $c_j^1,c_j^2,c_j^3$ are in $S$. If $c_j^2 \in S$, we must have that $c_j^3 \in S$, otherwise, $c_j^2$ gets a strictly worse allocation in $B$. If $c_j^3 \in S$, we must have at least one of the agents $c_j^2$ or $c_j^1 \in S$, otherwise she gets a strictly worse bundle in $B$ than in $A$. If only $c_j^3,c_j^2 \in S$, and $c_j^1 \notin S$, neither agent can strictly improve over their allocation in $A$. Therefore, we must have $c_j^1 \in S$.

\vspace{2mm}
\noindent \textendash~Let $[A(c_j^1)]_1 = [O(1_i^j)]_1$ (or $0_i^j$) and $[B(c_j^1)]_1 = [O(1_{i^*}^j)]_1$ (or $0_{i^*}^j$). Then agent $1_{i^*}^j$ (or $0_{i^*}^j$) corresponds to a literal $l_{i^*}$ in clause $c_j$ that is ranked weakly on top of $l_i$ according to $W$. Otherwise, $c_j^1$ is strictly worse off in $B$ than in $A$ according to the preferences in the construction, a contradiction to our assumption that $B$ is improving for every agent in $S$.

\vspace{2mm}
\noindent \textendash~If any one of $c_j^1 \in S$, then every agent $c_j^1 \in S$. Every agent $c_j^1$ must be assigned $[B(c_j^1)]_2 = [O(c_{j+1}^1)]_2$, otherwise she is strictly worse off in $B$. This also implies that $c_{j+1}^1 \in S$ (if $j=k$, replace $j+1$ with $1$).

\vspace{2mm}
\noindent \textendash~If $1_i^j \in S$ (or $0_i^j$), then $1_i^{j+1} \in S$ ($0_i^{j+1}$) and $1_i^{j-1} \in S$ ($0_i^{j-1}$). For agent $1_i^j$ to weakly improve her allocation in $B$ over her allocation in $A$, we must have $[B(1_i^j)]_2 = [O(1_i^{j+1})]_2$ (if $j = k_i$, replace $1_i^{j+1}$ with $1_{i+1}$). Otherwise, $1_i^j$ is strictly worse off in $B$. 

Some agent in $S$ must receive the item $[O(1_i^j)]$. Let agent $a \in S, a \neq 1_i^{j-1}$ be assigned item $[O(1_i^j)]$ in $B$, then by her reported preferences, agent $a$ is strictly worse off with her allocation in $B$ compared to her allocation in $A$, a contradiction to our assumption that $B$ is improving for every agent in $S$.

\vspace{2mm}
\noindent \textendash~For some pair $j,j'$, we cannot have both $1_i^j,0_i^{j'} \in S$. Suppose we have both $1_i^j$ and $0_i^{j'}$ in $S$, then we must have both $1_i^{k_i}$ and $0_i^{k_i} \in S$. However both of the agents $1_i^{k_i}$ and $0_i^{k_i}$ must be assigned the single item $[O(x_{i+1})]_2$ in $B$ to weakly improve over their allocation in $A$. Therefore one of the two agents must be strictly worse off in $B$, a contradiction to our assumption that allocations in $B$ are weakly preferred over $A$ by every agent in $S$.

\vspace{2mm}
\noindent \textendash~We have shown that $S$ must include all the agents $c_j^1$ corresponding to clauses $c_j$, for every $i$ either all of the agents $1_i^j$ or all of the agents $0_i^j$ and that for every $c_j$, an agent $1_i^j$ or $0_i^j$ corresponding to a valuation of a variable $x_i$ that satisfies $c_j$.

\vspace{2mm}
\noindent \textendash~Let $\psi$ be the valuation that is constructed as follows: for every clause $c_j$, if $c_j^1$ receives a type-$1$ item endowed to agent $1_i^j$, we set $\psi_i=1$. Note that every $c_j$ receives such an item and this corresponds to a valuation of the variable $x_i$ that satisfies $c_j$, and that once a variable takes a value, it is never changed again i.e. the value of $\psi_i$ is always consistent with a value that satisfies all clauses $c_j$ considered previously. Therefore, the values of the variables in $\psi$ constitutes a satisfying valuation. 

\vspace{2mm}
\noindent \textendash~Now, consider the allocation $B$ where the pairs of agents $c_j^1,1_i^j$ (or $0_i^j$) exchanging items of type $1$ is the same as in the allocation $A$. Then $B$ is not a strictly improving allocation over $A$ for any agent in $S$, a contradiction to our assumption on the allocation $B$. 

Let $i^*$ be the smallest value such that $\phi_{i^*} \neq \psi_{i^*}$. If $\psi_{i^*} = 1$, then $\psi$ is ranked over $\phi$ by $Q$, a contradiction to our assumption that $\phi$ is the highest ranked satisfying valuation. 

If $\psi_{i^*} = 0$, (then $\phi_{i^*} = 1$), then consider the agent $c_j$ who was assigned the type-$1$ item of $1_{i^*}^j$ in $A$. By the construction of $A$, $c_j$ must either receive the type-$1$ item of an agent $1_{\hat{i}}^j$ in $B$ where $\hat{i} > i^*$, or of an agent $0_{\hat{i}}^j$ in $B$. In either case, $c_j$ is strictly worse off in $B$ compared to $A$ since it corresponds to a lower ranked literal according to $W$. This is a contradiction to our assumption that every agent is weakly better off with their allocations in $B$ compared to $A$.
\end{proof}

\begin{myclaim}($\Leftarrow$)
If $J$ is a Yes instance, $I$ is a Yes instance.
\end{myclaim}
\begin{proof}
Suppose $A$ is a core allocation.

\vspace{2mm}
\noindent \textendash~Every agent $c_j^1$ must receive an item of type $1$ from an agent $1_i^j$ (or $0_i^j$) where the corresponding literal of $x_i$ satisfies clause $c_j$ i.e. we must have that $[A(c_j^1)]_1 = [O(1_i^j)]_1$ (similarly $[O(0_i^j)]_1$). 

For the sake of contradiction let $c_j^1$ not receive an item from an agent $1_i^j$ (or $0_i^j$). Then there are two cases: either $c_j^1,c_j^2,c_j^3$ are assigned each others' items or one of them is assigned the initial endowment of some other agent. If $c_j^1,c_j^2,c_j^3$ are only assigned each others' initial endowments, then by our construction of the preferences of agents $c_j^1,c_j^2,c_j^3$  by adapting Example~\ref{eg:Konishi} there is always a blocking coalition to every such allocation, a contradiction to our assumption that $A$ is in the core. If $c_j^2$ or $c_j^3$ is assigned the initial endowment of an agent other than $c_j^1,c_j^2$ or $c_j^3$, then the resulting allocation is not individually rational by the construction of agents' preferences, a contradiction. If $c_j^1$ is assigned the initial endowment of an agent other than an agent $1_i^j$ or $0_i^j$ corresponding to a satisfying assignment for the clause $c_j$, then again the resulting allocation is not individually rational.

\vspace{2mm}
\noindent \textendash~If $1_i^j$ (or $0_i^j$) exchanges her item type-$1$ item with an agent $c_j^1$, then none of the agents $0_i^{j'}$ (similarly, $1_i^{j'}$) can exchange items with some other agent $c_{j'}^1$.

For the sake of contradiction, let $1_i^j$ and $0_i^{j'}$ exchange their items of type $1$ with agents $c_j^1$ and $c_{j'}^1$ i.e they get $[A(1_i^{j})]_1 = [O(c_{j}^1)]_1$ and $[A(0_i^{j'})]_1 = [O(c_{j'}^1)]_1$ respectively. Then agents $1_i^{j}$ and $0_i^{j'}$ must get the assignments $[A(1_i^j)]_2 = [O(1^{j+1})]_2$ and $[A(0_i^{j'})]_2 = [O(0_i^{j'+1})]_2$ respectively. 

If $1_i^j \in S$ (or similarly, $0_i^j$), and $[O(1_i^j)]_2$ is assigned to another agent, then we must have that $[A(1_i^j)]_2 = [O(1_i^{j+1})]_2$, otherwise the allocation $A$ is not individually rational. 

By an inductive argument, we must have that the agents $1_i^{k_i}$ and $0_i^{k_i}$ must both be assigned the single indivisible item $[O(x_{i+1})]_2$ for $A$ to be individually rational, a contradiction to our assumption that $A$ is in the core since every core allocation must also be individually rational.

\vspace{2mm}
\noindent \textendash~Each agent $x_i$ receives either one of $[O(1_i^1)]_2$ or $[O(0_i^1)]_2$ since an agent can only receive a single item of a given type. By our argument immediately above, if $z_i$ receives $[O(1_i^1)]_2$ in $A$, then only agents $1_i^j$ can exchange items with agents $c_j^1$ and none of the agents $0_i^j$ can exchange items with an agent $c_j^1$.

\vspace{2mm}
\noindent \textendash~Now, consider the construction of a satisfying valuation $\phi$ as follows: For every clause $c_j$, if $[A(c_j^1)]_1 = [O(1_i^j)]_1$ (similarly, $0_i^j$), set $\phi_i=1$ which must be an assignment to the variable $x_i$ that satisfies clause $c_j$. We have already shown that such an assignment exists for every clause $c_j$. Further, once $\phi_i$ is set to a value, it can never change as we have already shown that if there is some $c_j$ such that $[A(c_j^1)]_1 = [O(1_i^j)]_1$ there is no $c_{j'}$ such that $[A(c_{j'}^1)]_1 = [O(0_i^{j'})]_1$. Therefore at every point where we verified the satisfaction of a clause $c_j$, the value of the variable $x_i$ in $\phi_i$ that satisfied $c_j$ is the same as the final value of $\phi_i$. Therefore, $\phi$ satisfies $F$.
\end{proof}

This completes the proof.
\end{proof}

\section{Summary and Future Work} We propose $\mttc$ for multi-type housing markets with lexicographic preferences, and prove that it satisfies many desirable axiomatic properties. 
There are  many future directions in mechanism design for multi-type housing markets. Are there good mechanisms when agents demand more than one item of some type? Can we design strategy-proof mechanisms under other assumptions about agents' preferences, such as LP-trees~\cite{Booth10:Learning}? What is the computational complexity of manipulation under $\mttc$? What if agents' preferences are partial orders such as CP-nets only?

\section{ Acknowledgments}
Research was sponsored by the Army Research Laboratory and was accomplished under Cooperative Agreement Number W911NF-09-2-0053 (the ARL Network Science CTA). The views and conclusions contained in this document are those of the authors and should not be interpreted as representing the official policies, either expressed or implied, of the Army Research Laboratory or the U.S. Government. The U.S. Government is authorized to reproduce and distribute reprints for Government purposes notwithstanding any copyright notation here on. Lirong Xia acknowledges National Science Foundation under grant IIS-1453542 for support.


\end{document}